\documentclass[10pt,journal,compsoc]{IEEEtran}
\ifCLASSOPTIONcompsoc
  \usepackage[nocompress]{cite}
\else
  \usepackage{cite}
\fi

\usepackage{bm}
\usepackage{amsmath}
\usepackage{epsf}
\usepackage{graphics}
\usepackage{ amssymb }
\usepackage[dvips]{graphicx}
\usepackage{epsfig}
\usepackage{cite}
\usepackage[linesnumbered,ruled,vlined]{algorithm2e}
\usepackage{colortbl}
\usepackage{color}
\newcommand{\overbar}[1]{\mkern 1.5mu\overline{\mkern-1.5mu#1\mkern-1.5mu}\mkern 1.5mu}
\usepackage{bm}
\usepackage{amsmath}
\usepackage{epsf}
\usepackage{graphics}
\usepackage{ amssymb }
\usepackage[dvips]{graphicx}
\usepackage{epsfig}
\usepackage{cite}
\usepackage[linesnumbered,ruled,vlined]{algorithm2e}
\usepackage{graphicx}
\usepackage{epsfig}
\usepackage{latexsym}
\usepackage{amsfonts}
\usepackage{here}
\usepackage{rawfonts}
\usepackage[utf8]{inputenc}
\usepackage[english]{babel}
\usepackage{amsmath}
\usepackage{amsfonts}
\usepackage{amssymb}
\usepackage{color} 
\usepackage{hyperref}
\usepackage{bm}
\usepackage{listings}
\usepackage{caption}
\usepackage{amssymb}
\usepackage{amsthm}
\usepackage{graphicx}
\usepackage{epstopdf}
\usepackage{listings}
\usepackage{float}
\usepackage{amsmath}
\usepackage{amssymb}
\usepackage{amsfonts}
\usepackage{epstopdf}
\usepackage[overload]{empheq}

\usepackage{multirow}
\usepackage{amscd}
\usepackage{mathrsfs}
\usepackage{graphicx}
\usepackage{color}
\usepackage{url}
\usepackage{bm}
\usepackage{setspace}
\usepackage{footnote}
\usepackage{xcolor}
\DeclareMathOperator*{\argmax}{arg\,max}

\DeclareMathOperator*{\argmin}{arg\,min}
\newcommand*{\Scale}[2][4]{\scalebox{#1}{$#2$}}%
\newtheorem{theorem}{Theorem}

\newtheorem{definition}{Definition}

\newtheorem{corollary}{Corollary}

\newtheorem{remark}{Remark}

\usepackage[]{algorithm2e}

\newcommand\ceil[1]{\lceil#1\rceil}

\renewcommand{\figurename}{Fig.}
\addto\captionsenglish{\renewcommand{\figurename}{Fig.}}

\usepackage{mathtools}

\usepackage{lipsum,graphicx,subcaption}
\newcommand{\GC}{GDCN}
\usepackage{diagbox}
\usepackage[protrusion=true,expansion=true]{microtype}
\usepackage[font=small]{caption}
\begin{document}

	\title{{\color{black}Power-Aware} Allocation of Graph Jobs in Geo-Distributed Cloud Networks}

	\author{Seyyedali Hosseinalipour,~\IEEEmembership{Student Member,~IEEE,} Anuj Nayak, and Huaiyu Dai,~\IEEEmembership{Fellow,~IEEE}
	\IEEEcompsocitemizethanks{\IEEEcompsocthanksitem{Seyyedali Hosseinalipour, Anuj Nayak, and Huaiyu Dai are with the Department
of Electrical and Computer Engineering, North Carolina State University, Raleigh,
NC, USA e-mail: {shossei3,aknayak,hdai}@ncsu.edu.}}}
	

\IEEEtitleabstractindextext{
	

	\begin{abstract}
		In the era of big-data, the jobs submitted to the clouds exhibit complicated structures represented by graphs, where the nodes denote the sub-tasks each of which can be accommodated at a \textit{slot} in a server, while the edges indicate the communication constraints among the sub-tasks. We develop a framework for efficient allocation of graph jobs in geo-distributed cloud networks ({\GC}s), explicitly considering the power consumption of the datacenters (DCs). We address the following two challenges arising in graph job allocation: i) the allocation problem belongs to NP-hard nonlinear integer programming; ii) the allocation requires solving the NP-complete sub-graph isomorphism problem, which is particularly cumbersome in large-scale {\GC}s. We develop a suite of efficient solutions for GDCNs of various scales. For small-scale {\GC}s, we propose an analytical approach based on convex programming. For medium-scale GDCNs, we develop a distributed allocation algorithm exploiting the processing power of DCs in parallel. Afterward, we provide a novel low-complexity (decentralized) sub-graph extraction method, based on which we introduce \textit{cloud crawlers} aiming to extract allocations of good potentials for large-scale GDCNs. Given these suggested strategies, we further investigate strategy selection under both fixed and adaptive DC pricing schemes, and propose an online learning algorithm for each. 
	\end{abstract}

	}

\maketitle

\IEEEdisplaynontitleabstractindextext

\IEEEpeerreviewmaketitle
	
   
\IEEEraisesectionheading{\section{Introduction}}
 \vspace{-0.5mm}
   \noindent \IEEEPARstart{R}{ecently}, the demand for \textit{big-data} processing has promoted the popularity of cloud computing platforms due to their reliability, scalability and security~\cite{ref:bigdataRise,ref:BigdataIOT,ref:BidDataService,big2,big3,big1}. Handling Big-data applications requires unique system-level design since these applications, more than often, cannot be processed via a single PC, server, or even a datacenter (DC). To this end, modern parallel and distributed processing systems (e.g., Apache/Twitter Storm~\cite{ref:ApapcheStorm}, GraphLab \cite{ref:HadoopWebsite}, IBM InfoSphere \cite{bookIBM}, MapReduce~\cite{onlinemapreduce}) are developed. In this work, we propose a framework for allocating big-data applications represented via \textit{graph jobs} in geo-distributed cloud networks ({\GC}s), explicitly considering the power consumption of the DCs. In the graph job model, each node denotes a sub-task of a big-data application while the edges impose the required communication constraints among the sub-tasks. 
    One of the common examples of processing graph jobs is receiving data from Twitter and counting the number of times a hashtag is mentioned, to keep an ordered list of the most commonly mentioned hashtags. Each step of the process is carried on in a so-called processing element, and it's these elements that enforce the separation of each logical step of the process (e.g. receiving updates, extracting hashtags, counting hashtags, ordering hastag count list) and allow the execution of the process on a distributed platform~\cite{Yahoos4}. In this context, a graph job is formed by viewing each element as a node and data exchange requirement between the elements as edges. As the sizes of the problem and graph jobs increase, one can imagine that a coalition of multiple DCs achieved through GDCNs is required for the execution of the graph jobs. 
 
   \subsection{Related Work}
There is a body of literature devoted to task and resource allocation in contemporary cloud networks, e.g.,~\cite{ref:allocation1,ref:allocation5,extref1,extref2,extref3,kaewpuang2013framework,ref:DGLB,survey,extra1,extra4}, {\color{black} where the topology of the graph job is not explicitly considered into their model}. In~\cite{ref:allocation1}, the task placement and resource allocation plan for embarrassingly parallel jobs, which are composed of a set of independent tasks, is addressed to minimize the job completion time. To this end, three algorithms named TaPRA, TaPRA-fast, and OnTaPRA are proposed, which significantly reduce the job execution time as compared to the state-of-the-art algorithms. In~\cite{ref:allocation5}, the multi-resource allocation problem in cloud computing systems is addressed through a mechanism called DRFH, where the resource pool is constructed from
a large number of heterogeneous servers containing various number of \textit{slots}. It is shown that DRFH leads to much higher resource utilization with considerably shorter job completion times.  
{\color{black}
In \cite{extref1}, the authors develop an online
job scheduling algorithm to distribute incoming workloads 
across multiple DCs targeting energy cost minimization with fairness
consideration subject to job delay requirements. They demonstrate that the solution of their online algorithm, which is solely based on current job queue lengths,
server availability and electricity prices, is close to the offline optimal performance
with future information.
In~\cite{extref2}, distribution of
the incoming workload among multiple DCs and adjustment of the service rates of the cloud servers are addressed aiming to reduce the power consumption cost. In
\cite{extref3}, the problem of directing the client requests to an appropriate DC efficiently and sending back the response packets to the client through one of the available links in the network is formulated as a workload management
optimization problem. To tackle the problem, the authors propose a distributed algorithm inspired by the
alternating direction method of multipliers.
} In~\cite{kaewpuang2013framework},
a resource allocation scheme is proposed resulting in efficient utilization of the resources
while increasing the revenue of the mobile cloud service providers. One of the pioneer works addressing resource allocation in {\GC}s {\color{black} considering the power consumption state of the DCs} is~\cite{ref:DGLB}, where a distributed algorithm, called DGLB, is proposed  for real-time
geographical load balancing. A good survey of the current state of the art is given in~\cite{survey}. Also, there is a body of literature utilizing swarm-based algorithms to perform the job allocation in cloud networks, e.g., \cite{extra1,extra4}. None of the above works has considered allocation of big-data
jobs composed of multiple sub-tasks requiring certain communication constraints among their sub-tasks. 

Allocation of big-data jobs represented by graph structures is a complicated process entailing more delicate analysis.  Among limited literature, references~\cite{zhou,ref:javad,ref:javad2} are most relevant, which focus on minimizing the cost incurred by utilizing the links among the adjacent DCs while neglecting the power consumption and the status of the utilized DCs. In~\cite{zhou}, a heuristic algorithm is developed to match the vertices of graph jobs to the idle slots of the cloud servers considering the cost of using the communication infrastructure of the network to handle the data flows among the sub-tasks. Using a similar system model in~\cite{ref:javad,ref:javad2}, the authors developed randomized algorithms for the same purpose. As compared to the heuristic approach of~\cite{zhou}, the authors of~\cite{ref:javad,ref:javad2} also demonstrate the optimality of their proposed algorithms through a theoretical approach. However, the algorithms used in these references are developed for a fixed network cost configuration, i.e., the cost of job execution using the same allocation strategy is fixed throughout the time. Also, as mentioned in~\cite{ref:nonpr}, the randomized algorithms proposed in~\cite{ref:javad,ref:javad2} suffer from long convergence time. In summary, the system model and the algorithms proposed in~\cite{zhou,ref:javad,ref:javad2} suffer from the following three limitations. i) The proposed algorithms are impractical in scenarios that the job allocation needs to be performed with respect to a time varying network cost configuration. ii) The proposed methods are impractical for large-scale networks. This is due to the fact that efficient handling of the NP-complete sub-graph isomorphism problem, which is a prerequisite to identify feasible allocations for graph jobs, is not directly addressed in these works (see Section~\ref{sec:Large}). iii) The proposed system models do not capture the power consumption of the utilized DCs. This is despite the fact that in {{\GC}}s, the execution cost is mainly determined by the real-time power consumption of the DCs~\cite{ref:powerDCsurvey}. Hence, an applicable allocation framework should be capable of fast allocation of incoming graph jobs to the {{\GC}}s considering the effect of allocation on the current DCs' power consumption state. Also, with the rapid growth in the size of cloud networks, adaptability to large-scale {\GC}s is a must for such a framework. These are the main motivations behind this work.

\subsection{Contributions}
The main goal of this paper is to provide a framework for graph job allocation in {\GC}s with various scales. Our main contributions can be summarized as follows:

\noindent 1) We formulate the problem of graph job allocation in {\GC}s considering the incurred power consumption on the cloud network. 

\noindent 2) We propose a centralized approach to solve the problem suitable for small-scale cloud networks. 

\noindent 3) We design a distributed algorithm for allocation of graph jobs in medium-scale {{\GC}}s, using the DCs' processing power in parallel. 

\noindent 4) For large-scale {{\GC}}s, given the huge size of the strategy set, and extremely slow convergence of the distributed algorithm, we introduce the idea of \textit{cloud crawling}. In particular, we propose a fast method to address the NP-complete \textit{sub-graph isomorphism problem}, which is one of the major challenges for graph job allocation in cloud networks. In this regard, we propose a novel low-complexity (decentralized) sub-graph isomorphism extraction algorithm for a cloud crawler to identify ``potentially good" strategies for customers while traversing a {\GC}.

\noindent 5) For large-scale {\GC}s, considering the suggested strategies of cloud crawlers, we find the best suggested strategies for the customers under adaptive and fixed pricing of the DCs in a distributed fashion. To this end, we model proxy agents' behavior in a {\GC}, based on which we propose two online learning algorithms inspired by the concept of ``regret" in the bandit problem~\cite{ref:gambler,ref:regretmatching}.

This paper is organized as follows. Section~\ref{sec:systModel} includes system model. Section~\ref{sec:Small} contains a sub-optimal approach for graph job allocation in small-scale {\GC}s. A distributed graph job allocation mechanism for medium-scale {\GC}s is presented in Section \ref{sec:Medium}. Cloud crawling along with online learning algorithms for large-scale {\GC}s are presented in Section~\ref{sec:Large}. Simulation results are given in section~\ref{sec:Simu}. Finally, Section~\ref{sec:concl} concludes the paper.

    \section{System Model}\label{sec:systModel}
    \noindent A {\GC} comprises various DCs connected through communication links. Inside each DC, there is a set of fully-connected cloud servers each consisting of multiple fully-connected \textit{slots}. Without loss of generality, we assume that all the cloud servers have the same number of slots. Each slot corresponds to the same bundle of processing resources which can be utilized independently. Since all the slots belonging to the same DC are fully-connected, we consider a DC as a collection of slots directly in our study.\footnote{The number of cloud servers does not play a major role in our study except in the energy consumption models.} It is assumed that a DC provider (DCP) is in charge of DC management. Abstracting each DC to a \textit{node} and a communication link between two DCs as an \textit{edge}, a {{\GC}} with $n_d$ DCs can be represented as a graph $G_D=(\mathcal{D},\mathcal{E}_D)$, where $\mathcal{D}=\{d^1,\cdots,d^{n_d}\}$ denotes the set of nodes and $\mathcal{E}_D$ represents the set of edges. Henceforth, $G_D$ is assumed to be \textit{connected}; however, due to the geographical constraints, $G_D$ may not be a complete graph.
    
    Let $\mathcal{S}^i=\{S^i_1,\cdots,S^i_{|\mathcal{S}^i|}\}$ denote the set of slots belonging to DC $d^i$. The existence of a connection between two DCs leads to the communication capability between all of their slots. Consequently, two slots are called \textit{adjacent} if and only if both belong to the same DC or there exists a link between their corresponding DCs. Let ${\color{black}\mathcal{E}_S}$ denote the set of edges between the adjacent slots, where $(S^i_k,S^j_m)\in {\color{black}\mathcal{E}_S}$ if and only if $i=j, \forall k\neq m$ or $(d^i,d^j)\in \mathcal{E}_D, \forall k,m$. We define the \textit{aggregated network graph} as $G=(\mathcal{V}_S,{\color{black}\mathcal{E}_S})$, where $\mathcal{V}_S=\cup_{i=1}^{n_d} \mathcal{S}^i$ and $|\mathcal{V}_S|= \sum_{i=1}^{n_d} |\mathcal{S}^i|$.

    Let $\mathcal{J}=\{Gjob_1, Gjob_2,\cdots,Gjob_J\}$, denote the set of all possible types of the graph jobs in the system, each of which is considered as a graph $Gjob_j=(\mathcal{V}_j,\mathcal{E}_j)$. Each node of a graph job requires one slot from a DC to get executed. It is assumed that $\mathcal{V}_j=\{v^1_j,\cdots,v^{n_j}_j\}$, and $\forall (m,n):1\leq m \neq n \leq n_j$, $\left(v^m_j,v^n_j\right) \in \mathcal{E}_j$ if and only if the nodes $v^m_j$ and $v^n_j$ need to be executed using two adjacent slots of the {{\GC}}. {\color{black} Similar to~\cite{ref:javad,ref:javad2,zhou}, we assume that allocation of all the nodes is required during the execution of the respective job.}
    
The system model is depicted in Fig.~\ref{diag:sysmodel}. For the small- and medium-scale {\GC}s, the {\GC} network is assumed to be in charge of finding adequate allocations for the incoming graph jobs from proxy agents (PAs)~(\hspace{-1.59mm} \cite{ref:agent,Ali:twoStage,Ali:option}), which act as trusted parties between the {\GC} and the customers. In these cases, each graph job is allocated through either a centralized controller or a distributed algorithm utilizing the communication infrastructure between the DCs (see Section~\ref{sec:Medium}). For large-scale {\GC}s, cloud crawlers are introduced to explore the {\GC} to provide a set of suggested strategies for the PAs. Afterward, PAs allocate their graph jobs with respect to the utility of the suggested strategies (see Section~\ref{sec:Large}). The following definitions are introduced to facilitate  our subsequent derivations.    
    \begin{definition}
    A \textbf{feasible mapping} between a $Gjob_j$ and the {{\GC}} is defined as a mapping $f_j:\mathcal{V}_j \mapsto \mathcal{V}_S$, which satisfies the communication constraints of the graph job. This implies that $\forall (m,n): 1\leq m \neq n \leq |\mathcal{V}_j|$, if $(v^m_j,v^{n}_j)\in \mathcal{E}_j$, then $\left(f_j(v^m_j),f_j(v^{n}_j)\right)\in {\color{black}\mathcal{E}_S}$. Let $\mathcal{F}_j = \{f^1_j, \cdots, f^{|\mathcal{F}_j|}_j\}$ denote the set of all feasible mappings for the $Gjob_j$.
    \end{definition}
    \begin{definition}
    For a $Gjob_j$, a \textbf{mapping vector} associated with a feasible mapping $f^k_j\in \mathcal{F}_j $ is defined as a vector $\mathbf{M}_j|_{f^k_j} =[m^1_j|_{f^k_j},\cdots,m^{n_d}_j|_{f^k_j}]\in (\mathbb{Z}^+ \cup \{0\})^{n_d}$, where $m^i_j|_{f^k_j}$ denotes the number of used slots from DC $d^i$. Mathematically, $m^i_j|_{f^k_j}=\sum_{l=1}^{|\mathcal{V}_j|} \mathbf{1}_{\{f^k_j(v^l_j)\in \mathcal{S}^i\}}$, where $\mathbf{1}_{\{.\}}$ represents the indicator function. Let $\mathcal{M}_j=\{ \mathbf{M}_j|_{f^1_j},\cdots, \mathbf{M}_j|_{f^{|\mathcal{F}_j|}_j} \}$ denote the set of all mapping vectors for the $Gjob_j$. 
    \end{definition}
    
    Finding a feasible allocation/mapping between a graph job and a {\GC} is similar to the \textit{sub-graph isomorphism problem} in graph theory~\cite{subgraphNPcomplete}. Some examples of feasible allocations for a graph job with three nodes considering a {\GC} with four DCs each consisting of four slots is depicted in Fig.~\ref{fig:iso}.
    
   {\color{black} Our aim is to allocate big-data driven applications, e.g., computation intensive big-data applications~\cite{zhou} or data streams~\cite{ref:javad,ref:javad2}, to {\GC}s. Due to the nature of these applications, the jobs usually stay in the system so long as they are not terminated. This work can be considered as a real-time allocation of graph jobs to the system, where we find the best currently possible assignment considering the current network status. Hence, we deliberately omit the time index from the following discussions.}    Inspired by~\cite{ref:powerInfo,ref:powerDCsurvey}, we model the power consumption upon utilizing $s$ slots of $d^i$ comprising $N^i$ cloud servers each with idle power consumption $P^i_{idle}$ as:
   \vspace{-2mm}
   \begin{equation}\label{eq:powerPrice0}
  \eta^i N^i \left( \sigma^i \left(\frac{s}{|\mathcal{S}^i|}\right)^{\alpha^i}+P^i_{idle} \right),  \alpha^i\geq 2.
   \end{equation}
   In this model, $\eta^i$ is the so-called \textit{Power Usage Effectiveness}, which is the ratio between the total power usage (including cooling, lights, UPS, etc.) and the power consumed by the IT-equipment of a DC, and $\sigma^i$ is chosen in such a way that $\sigma^i+P^i_{idle}$ determines the peak power consumption of a cloud sever $P^i_{max}$ inside $d^i$. Also, $\alpha^i$ is a DC-related constant. Subsequently, we define the incurred cost of executing a graph job with type $j$ allocated according to the feasible mapping vector $\mathbf{M}_j=[m^1_j,\cdots,m^{n_d}_j]$ as follows:
   \vspace{-2mm}
    \begin{equation}\label{eq:powerPrice}
   \Scale[0.98]{\displaystyle \sum_{i=1}^{n_d}    \xi^i \eta^i N^i \left( \sigma^i \left(\frac{{L}^i+m^i_{j}}{|\mathcal{S}^i|}\right)^{\alpha^i}+P^i_{idle} \right) + \sum_{i=1}^{n_d} \xi^i \nu^i m^i_j,}
    \end{equation}
   where ${L}^i$ is the original load of DC $d^i$, $\nu^i$ indicates the I/O incurred power of using the communication infrastructure of DC $d^i$ per slot, and $ \xi^i$ is the ratio between the cost and power consumption, which is dependent on the DC's location and infrastructure design. The I/O cost is considered to be proportional to the number of used slots since the data generated at each DC is correlated with that number, and that data should be exchanged using the I/O infrastructure either among adjacent DCs or between DCs and the users. {\color{black} Note than Eq.~\eqref{eq:powerPrice0} and Eq.~\eqref{eq:powerPrice} do not capture the order of the slots used in each DC and assume that the utilization of a slot from every server in a DC requires the same power consumption. However, in reality some servers might be in the idle mode, which need more power to boot and execute the process. Since each DC may contain tens of servers, considering the status of each server increases the dimension of the problem significantly, which makes the problem intractable even in small-scale GDCNs. Also, obtaining the status of all the servers in all the DCs is challenging.  Addressing these issues is out of the scope of this paper and left as a future work. In this paper, we assume that after allocation of the graph jobs to the GDCN and sending the information to the respective DCs, each DC manager makes an internal decision about the effective usage of the servers' slots considering the status of the servers. }
\begin{figure*}[t]
		\minipage{13cm}
		\vspace{-2mm}
		\includegraphics[width=130mm,height=70mm]{Figures/cloud_gdcn7.pdf}
		\caption{System model for graph job allocation in {\GC}s with various scales.}
		\label{diag:sysmodel}
		\endminipage\hfill
		\vspace{-5mm}
				\quad
		        \minipage{4cm}
		\includegraphics[width=1.55in,height=55mm]{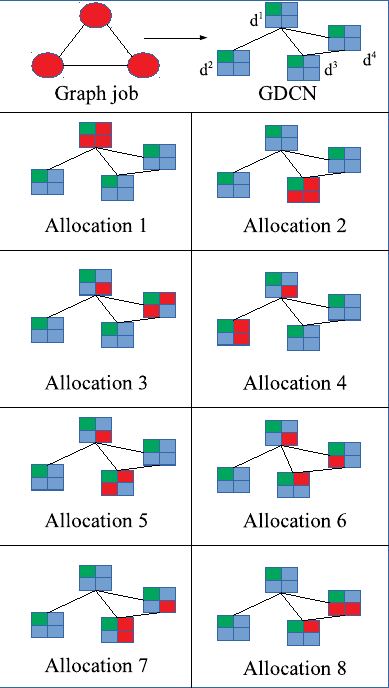}
		\caption{Examples of graph job allocation. The green (blue) color denotes busy (idle) slots. The red color indicates the utilized slots upon allocation.}
		 \label{fig:iso}
		\endminipage
	\end{figure*}

\subsection{Problem Formulation}
Our goal is to find an allocation for each arriving graph job to minimize the total incurred cost on the network. Due to the inherent relation between the cost and loads of the DCs, minimizing the cost is coupled with balancing the loads of the DCs. In a {\GC}, let $N_j$ denote the number of $Gjob_j \in \mathcal{J}$ in the system demanded for execution. Let $\mathbf{\overbar{M}}_{j}$ denote the matrix of mapping vectors of these graph jobs defined as follows:

 {\small 
      \[
\mathbf{\overbar{M}}_{j}=
  \begin{bmatrix}
 \mathbf{M}_{j,(1)},\mathbf{M}_{j,(2)},\cdots,\mathbf{M}_{j,({N_j})}
  \end{bmatrix},
  \;\; \forall j\in \{1,\cdots, J\},
  \]
  }
         {\small 
      \[
\mathbf{M}_{j,(i)}=
  \begin{bmatrix}
    m^1_{j,(i)} ,
    m^2_{j,(i)} ,
    \cdots ,
    m^{n_d}_{j,(i)}
  \end{bmatrix}^\top,
  \; \forall i\in \{1,\cdots, N_j\}.
\]
    }
    
We formulate the \textit{optimal graph job allocation} as the following optimization problem ($\mathcal{P}1$): 
\vspace{-2mm}
        \begin{flalign}
    &\hspace{5mm}[\mathbf{\overbar{M}}^*_1,\mathbf{\overbar{M}}^*_2,\cdots,\mathbf{\overbar{M}}^*_J]=\nonumber\\
    &\hspace{3mm}\argmin_{[\mathbf{\overbar{M}}_1,\mathbf{\overbar{M}}_2,\cdots,\mathbf{\overbar{M}}_J]} \sum_{i=1}^{n_d}    \xi^i \eta^i N^i  \Bigg( \sigma^i \left(\frac{L^{i}+\sum_{j=1}^{J} \sum_{k=1}^{N_j} m^i_{j,(k)}}{|\mathcal{S}^i|}\right)^{\alpha^i}\nonumber \\
    &\hspace{5mm}+P^i_{idle}  \Bigg) + \sum_{i=1}^{n_d} \sum_{j=1}^{J} \sum_{k=1}^{N_j}  \xi^i \nu^i m^i_{j,(k)}  \\
     &\hspace{5mm}\textrm{s.t.}\nonumber\\
     &\hspace{5mm}\sum_{j=1}^{J} \sum_{k=1}^{N_j} m^i_{j,(k)}  \leq |\mathcal{S}|^{i}-L^{i},\;\;\forall i\in\{1,\cdots,n_d\},\label{eq:const1}  \\
     &\hspace{5mm}\mathbf{M}_{j,(i)} \in \mathcal{M}_j, \;\;\forall j\in \{1,\cdots,J\}, \forall i \in\{1,\cdots,N_j\}\label{eq:const2} .
    \end{flalign}
           In $\mathcal{P}1$, the objective function is the total incurred cost of execution, {\color{black}the first condition given by \eqref{eq:const1} ensures the stability of the DCs, and the second constraint given by \eqref{eq:const2} guarantees the feasibility of the assignment.} There are two main difficulties in obtaining the solution: i) Identifying the feasible mappings ($\mathcal{M}_j$-s) requires solving the sub-graph isomorphism problem between the graph jobs' topology and the aggregated network graph, which is categorized as NP-complete~\cite{subgraphNPcomplete}. Hence, we only assume the knowledge of $\mathcal{M}_j$-s in the small- and medium-scale {\GC}s. In the large-scale {\GC}s, we propose a low-complexity decentralized approach to extract isomorphic sub-graphs to a graph job and implement it in our proposed cloud crawlers. ii) $\mathcal{P}1$ is a nonlinear integer programming problem, which is known to be NP-hard. In small- and medium-scale {\GC}s, we tackle this problem considering a convex relaxed version of it. However, for large-scale {\GC}s,  we find a ``potentially good" subset of feasible mappings as the cloud crawlers traverse the network. Afterward, the strategy selection is carried out using the computing power of the PAs in a decentralized fashion.
 {\color{black}
 \begin{remark}
 Considering the possibility of link outages and security preferences, the users may prefer utilizing fewer DCs during the job execution. Since these situations are more likely to happen in large-scale GDCNs, we incorporate this tendency to the utility function of the users, i.e., Eq.~\eqref{eq:util2} and~Eq.~\eqref{eq:util1}, described in Section~\ref{sec:Large}.
 \end{remark}}
       \begin{table}    
{\footnotesize
\caption{Major notations.}
    \begin{center}
    \vspace{-2mm}
\begin{tabular}{ |c |c|} 
 \hline
 Symbol& Definition\\ \hline
 $G_D$ & The {{\GC}} graph  \\ \hline
 $\mathcal{D}$ & Set of DCs in the {{\GC}} \\ \hline
 $d^i$ & The DC with index $i$ \\ \hline
 $n_d$ & Number of DCs in the {{\GC}}  \\ \hline
  $\mathcal{S}^i$ & Set of slots of DC $d^i$  \\ \hline
 $G$ & Aggregated graph of the {{\GC}}  \\ \hline
   $\mathcal{V}_S$& Set of slots of the entire {\GC}\\ \hline
{\color{black}   $\mathcal{E}_D$}& {\color{black}Set of edges between adjacent DCs of a {\GC}}\\ \hline
  ${\color{black}\mathcal{E}_S}$& Set of edges between adjacent slots of DCs in a {\GC}\\ \hline
  $\mathcal{J}$ & Set of graph jobs in the system  \\ \hline
  $J$ & Number of different types of jobs in the system  \\ \hline
 $Gjob_j$ & Associated graph to the graph job with type $j$ \\ \hline
   $N_j$ & Number of jobs with type $j$ in the system   \\ \hline
    $\mathcal{V}_j$& Set of nodes of the graph job with type $j$\\ \hline
      $\mathcal{E}_j$& Set of edges of the graph job with type $j$\\ \hline
  $L^i$ & Load of DC $d^i$\\ \hline
  $N^i$ & Number of cloud servers in DC $d^i$   \\ \hline
    $\mathcal{M}_j$& Set of all the mapping vectors for $Gjob_j$ \\ \hline
       $\mathcal{P}$& Set of PAs in the system\\ \hline 
       $\mathcal{SA}_j$& Set of cloud crawler's suggested strategies for $Gjob_j$\\ \hline 
       $\mathbf{p}_{j,(m)}$& Probability of selection of strategy $m\in\mathcal{SA}_j$\\ \hline 
\end{tabular}
\end{center}
}
\end{table}
    \section{Graph Job Allocation in Small-Scale \GC s: Centralized Approach}\label{sec:Small}
 \noindent Solving $\mathcal{P}1$ requires solving an integer programming problem in $n_d{\sum_{j=1}^{J}} N_j$ dimensions. For a small {\GC} with three types of graph jobs ($J=3$), $5$ DCs ($n_d=5$), and $100$ graph jobs of each type in the system, the dimension of the solution becomes $1500$ rendering the computations impractical. To alleviate this issue, we solve $\mathcal{P}1$ in a sequential manner for available graph jobs in the system. In our approach, at each stage, the best allocation is obtained for one graph job while neglecting the presence of the rest. Afterward, the graph job is allocated to the {\GC} and the loads of the utilized DCs are updated. As a result, at each stage, the dimension of the solution is $n_d$ ($5$ in the above example). For a $Gjob_j\in\mathcal{J}$, let the available graph jobs be indexed from $1$ to $N_j$ according to their execution order, where preferred customers can be prioritized in practice. For a graph job with type $j$ with index $k$, we reformulate $\mathcal{P}1$ as ($\mathcal{P}2$):
        \begin{flalign}
        &\hspace{1.9mm}\mathbf{M}^*_{j,(k)}\hspace{-1mm}=\hspace{-0.5mm}\argmin_{\mathbf{M}_{j,(k)}} \sum_{i=1}^{n_d}    \xi^i \eta^i N^i \hspace{-1mm} \left( \sigma^i \left(\frac{L^{i}+ m^i_{j,(k)}}{|\mathcal{S}^i|}\right)^{\alpha^i}\hspace{-4mm}+P^i_{idle}  \right) \nonumber\\
    &\hspace{1.9mm}+ \sum_{i=1}^{n_d}  \xi^i \nu^i m^i_{j,(k)} \label{eq:int1}\\
     &\hspace{1.9mm}\textrm{s.t.}\nonumber\\
     &\hspace{1.9mm} m^i_{j,(k)}  \leq |\mathcal{S}^i|-L^{i}, \;\;\; \forall i\in\{1,2,\cdots,n_d\}, \label{eq:insideServerload}\\
     &\hspace{1.9mm}\mathbf{M}_{j,(k)}\in \mathcal{M}_j, \label{eq:int2}
    \end{flalign}
    where $L^{i}$ denotes the updated load of DC $d^i$ after the previous graph job allocation. The last constraint in $\mathcal{P}2$ forces the solution to be discrete making the derivation of a tractable solution impossible. In the following, we relax this constraint and provide a tractable method to derive the solution in the set of feasible points.
 For the moment, we consider $\mathbf{M}_{j,(k)} \in \mathbb{\left(R^+\right)}^{n_d},\; \forall j,k$. {\color{black} We define $\mathcal{P}3$ as the following optimization problem with the same objective function as $\mathcal{P}2$, in which the constraint given by Eq.~\eqref{eq:int2} is relaxed and represented as two constraints ($\mathcal{P}3$):
      \begin{flalign}
      &\hspace{1.9mm}\mathbf{M}^*_{j,(k)}\hspace{-1mm}=\hspace{-0.5mm}\argmin_{\mathbf{M}_{j,(k)}} \sum_{i=1}^{n_d}    \xi^i \eta^i N^i \hspace{-1mm} \left( \sigma^i \left(\frac{L^{i}+ m^i_{j,(k)}}{|\mathcal{S}^i|}\right)^{\alpha^i}\hspace{-4mm}+P^i_{idle}  \right) \nonumber\\
    &\hspace{1.9mm}+ \sum_{i=1}^{n_d}  \xi^i \nu^i m^i_{j,(k)} \\
     &\hspace{1.9mm}\textrm{s.t.}~\eqref{eq:insideServerload},\nonumber\\
     &\hspace{3mm}\sum_{i=1}^{n_d} m^i_{j,(k)}  =|\mathcal{V}_j|,\label{eq:insidefeasible}\\
     & \hspace{3mm} m^i_{j,(k)}\geq 0, \;\; \forall i\in \{1,2,\cdots,n_d\},\label{eq:rel2}
    \end{flalign}
 where Eq.~\eqref{eq:insidefeasible} ensures the assignment of all the nodes of the graph job to the {\GC}, and Eq.~\eqref{eq:rel2} guarantees the practicality of the solution.} It is easy to verify that $\mathcal{P}3$ is a convex optimization problem. We use the \textit{Lagrangian dual decomposition} method~\cite{boyd} to solve this problem. 
 Let $\bm{\lambda}=[\lambda^1,\lambda^2,\cdots,\lambda^{n_d}]$, $\gamma$, and $\bm{\Lambda}=[\Lambda^1,\Lambda^2,\cdots,\Lambda^{n_d}]$ denote the Lagrangian multipliers associated with the first, the second, and the third constraint, respectively. The \textit{Lagrangian function} associated with $\mathcal{P}3$ is then given by:
 \begin{flalign}
 &\hspace{5mm}L(\mathbf{M}_{j,(k)},\bm{\lambda},\gamma,\bm{\Lambda})= -\sum_{i=1}^{n_d} \Lambda^i m^i_{j,(k)}+\nonumber\\ 
 &\hspace{5mm}\sum_{i=1}^{n_d}   \xi^i \eta^i N^i  \left( \sigma^i \left(\frac{L^{i}+ m^i_{j,(k)}}{|\mathcal{S}^i|}\right)^{\alpha^i}\hspace{-4mm}+P^i_{idle}  \right) +\sum_{i=1}^{n_d} \xi^i \nu^i m^i_{j,(k)} \nonumber \\
&\hspace{4mm}+ \sum_{i=1}^{n_d} \lambda^i \left(m^i_{j,(k)} - |\mathcal{S}^i|+L^{i}\right) + \gamma \left( \sum_{i=1}^{n_d} m^i_{j,(k)}  - |\mathcal{V}_j|\right).
 \end{flalign}
 The corresponding \textit{dual function} of $\mathcal{P}3$ is given by:
  \begin{align}\label{eq:dualfunc}
  D(\bm{\lambda},\gamma,\bm{\Lambda})= \min_{\mathbf{M}_{j,(k)}}{L(\mathbf{M}_{j,(k)},\bm{\lambda},\gamma,\bm{\Lambda})}.
  \end{align}
 Finally, the \textit{dual problem} can be written as ($\mathcal{P}4$):
 \begin{align}\label{ref:dualProb}
  \max_{\bm{\lambda},\bm{\Lambda}\in (\mathbb{R}^+)^{n_d},\gamma \in \mathbb{R}} D(\bm{\lambda},\gamma,\bm{\Lambda}).
 \end{align}
 $\mathcal{P}3$ is a convex optimization problem with differentiable affine constraints; hence, it satisfies the \textit{constraint qualifications} implying a zero duality gap. As a result, the solution of $\mathcal{P}3$ coincides with the solution of $\mathcal{P}4$. It can be verified that the minimum of the Lagrangian function occurs at the following point:
 \begin{equation}\label{eq:ms}
\hspace{-1mm}\Scale[0.99]{{m^i_{j,(k)}}^* = \left(\frac{\Lambda^i - \lambda^i - \gamma -\xi^i \nu^i}{\frac{  \xi^i \eta^i N^i  \sigma^i \alpha^i}{(|\mathcal{S}^i|)^{\alpha^i}}}\right)^{\frac{1}{\alpha^i-1}}\hspace{-3mm} -L^i, \;\forall i\in \{1,\cdots,n_d \}.}
 \end{equation}
 By replacing this in the Lagrangian function, the dual function is given by: $D(\bm{\lambda},\gamma,\bm{\Lambda})=L(\mathbf{M}^*_{j,(k)},\bm{\lambda},\gamma,\bm{\Lambda})$, where $\small 
\mathbf{M}^*_{j,(k)}=[
     {m^1_{j,(k)}}^* ,
     {m^2_{j,(k)}}^* ,
    \cdots ,
     {m^{n_d}_{j,(k)}}^*
  ]^\top$.
 The optimal Lagrangian multipliers can be obtained by solving the dual problem given by:
 \begin{equation}\label{eq:nabla}
 \nabla D(\bm{\lambda},\gamma,\bm{\Lambda})|_{(\bm{\lambda}^*, \gamma^*,\bm{\Lambda}^*)}=0.
 \end{equation}
 Given the solution of Eq.~\eqref{eq:nabla}, the optimal allocation in $\mathbb{\left(R^+\right)}^{n_d}$ is given by $\mathbf{M}^*_{j,(k)}|_{(\bm{\lambda}^*, \gamma^*,\bm{\Lambda}^*)}$.
The solutions of Eq. $\eqref{eq:nabla}$ can be derived via the iterative \textit{gradient ascent algorithm}~\cite{boyd}.

Let $\widetilde{\mathbf{M}_{j,(k)}^*}=[\widetilde{{m^{1}_{j,(k)}}^*},\cdots,\widetilde{{m^{n_d}_{j,(k)}}^*}]^\top$ denote the derived solution in the continuous space, we obtain the solution of $\mathcal{P}2$ by solving the following weighted mean-square problem:\footnote{Instead of solving the mean-square problem, the k-d tree data structure~\cite{bentley1975multidimensional} can be used to find the closest feasible allocation in average complexity of $O(\log(|\mathcal{M}_{j}|))$, where $|\mathcal{M}_{j}|$ is the number of feasible allocations.}
\begin{equation}\label{eq:weightedmeanSquare}
\begin{aligned}
\mathbf{M}^*_{j,(k)}= \argmin_{\mathbf{M}_{j,(k)}\in \mathcal{M}_{j}} \sum_{i=1}^{n_d} w_i \left(m^i_{j,(k)}-\widetilde{{m^i_{j,(k)}}^*}\right)^2,
\end{aligned}
\end{equation} 
where $w_.$-s are the design parameters, which can be tuned to impose a certain tendency toward utilizing specific DCs.

So far, to derive the above solution, it is necessary to have a powerful centralized processor with global knowledge about the state of all the DCs. This is due to the inherent updating mechanism of the gradient ascent method~\cite{boyd}, in which iterative update of each Lagrangian multiplier requires global knowledge of the current values of the other Lagrangian multipliers and the DCs' loads. Obtaining this knowledge may not be feasible for a given {\GC} with more than a few DCs. Moreover, multiple powerful backup processors may be needed to avoid the interruption of the allocation process in situations such as overheating of the centralized processor. In the following section, we design a distributed algorithm using the processing power of the DCs in parallel to resolve the above concerns.

 \section{Graph Job Allocation in Medium-Scale {\GC}s: Decentralized Approach with DCs in Charge of Job Allocation}\label{sec:Medium}
\noindent The described dual problem in Eq.~\eqref{ref:dualProb}, given the result of Eq.~\eqref{eq:ms}, can be written as follows:
\begin{align}\label{eq:needdistributed}
&\max_{\lambda^i \in \mathbb{R}^+,\gamma \in \mathbb{R},\Lambda^i\in \mathbb{R}^+}\sum_{i=1}^{n_d} D^i({\lambda}^i,\gamma,{\Lambda}^i),
\end{align}
where
\begin{align}
&\hspace{-3.5mm}D^i({\lambda}^i,\gamma,{\Lambda}^i)=    \xi^i \eta^i N^i  \left( \sigma^i \left(\frac{L^{i}+ {m^i_{j,(k)}}^*}{|\mathcal{S}^i|}\right)^{\alpha^i}\hspace{-4mm}+P^i_{idle}  \right) \nonumber\\
&\hspace{-3.5mm}+  \xi^i \nu^i {m^i_{j,(k)}}^* + \lambda^i \left({m^i_{j,(k)}}^*  
- |\mathcal{S}^i|+L^{i}\right)\nonumber\\  
&\hspace{-3.5mm}+\hspace{-1mm} \gamma \hspace{-0.5mm}\left({m^i_{j,(k)}}^* \hspace{-1mm} - |\mathcal{V}_j|/n_d\right) \hspace{-1mm}-\hspace{-1mm} \Lambda^i {m^i_{j,(k)}}^*,\;\;\forall i\in \{1,\cdots,n_d \}.
\end{align}
In Eq.~\eqref{eq:needdistributed}, each term can be associated with a DC. For $d^i$, there are two private (local) variables $\lambda^i,\Lambda^i$ and a public (global) variable $\gamma$, which is identical for all the DCs. Due to the existence of this public variable, the objective function cannot be directly written as a sum of separable functions. In the following, we propose a distributed algorithm deploying local exchange of information among adjacent DCs to obtain a unified value for the public variable across the network.
\subsection{Consensus-based Graph Job Allocation}
We propose the consensus-based distributed graph job allocation (CDGA) algorithm consisting of two steps to find the solution of Eq.~\eqref{eq:needdistributed}: i) updating the local variables at each DC, ii) updating the global variable via forming a consensus among DCs. We consider each term of Eq.~\eqref{eq:needdistributed} as a (hypothetically) separate term and rewrite the problem as a summation of separable functions, with $\gamma$ replaced by $\gamma^i$ in $D^i(.,.,.)$:
\begin{equation}\label{eq:full_dist_1}
\max_{\lambda^i \in \mathbb{R}^+,\gamma^i \in \mathbb{R},\Lambda^i\in \mathbb{R}^+} \sum_{i=1}^{n_d} D^i({\lambda}^i,{\gamma}^i,{\Lambda}^i).
\end{equation} 
At each iteration of the CDGA algorithm, each DC first derives the value of the following variables locally using the gradient ascent method:
\begin{equation}\label{eq:full_dist_2}
\begin{aligned}
&\hspace{-3mm}\lambda^i(k+1)=\lambda^i(k)+ c_{\lambda} (\nabla_{\lambda^i} D^i({\lambda}^i(k),{\gamma}^i(k),{\Lambda}^i(k))),\\
&\hspace{-3mm}\gamma'^i(k+1)= \gamma^i(k) + c_{\gamma} (\nabla_{\gamma^i} D^i({\lambda}^i(k),{\gamma}^i(k),{\Lambda}^i(k))),\\
&\hspace{-3mm}\Lambda^i(k+1)= \Lambda^i(k)     + c_{\Lambda} (\nabla_{\Lambda^i} D^i({\lambda}^i(k),{\gamma}^i(k),{\Lambda}^i(k))),
\end{aligned}
\end{equation}
where $c_{.}$-s are the corresponding step-sizes and $\gamma'^i$ is a local variable. Afterward, the local copies of the global variable ($\gamma^i$-s) are derived by employing the consensus-based gradient ascent method~\cite{ref:DistCens}:
\begin{align}\label{eq:full_dist_3}
\gamma^i(k+1)=\sum_{j=1}^{n_d} \Big(\mathbf{W}^\Phi\Big)_{ij} \gamma'^j(k) ,
\end{align}
where $\mathbf{W}=\mathbf{I}-\epsilon \mathbf{L}(G_D)$, with $\mathbf{L}(G_D)$ the Laplacian matrix of $G_D$ and $\epsilon \in (0,1)$, and $\Phi\in \mathbb{N}$ denotes the number of performed consensus iterations among the adjacent DCs. In this method, the adjacent DCs perform $\Phi$ consensus iteration with local exchange of $\gamma'$-s before updating $\gamma$. The pseudo-code of the CDGA algorithm is given in Algorithm~\ref{alg:fulldist}. Since the solution is found in the continuous space, similar to Section~\ref{sec:Small}, the last stage of the algorithm is obtaining the solution in the feasible set of allocations. This step requires a centralized processor with the knowledge of the feasible solutions. Nevertheless, as compared to the centralized approach (Section~\ref{sec:Small}), the centralized processor is no longer in charge of deriving the optimal allocations for each graph job.

 \begin{algorithm}[t]
 {\footnotesize
 	\caption{CDGA: Consensus-based distributed graph job allocation}\label{alg:fulldist}
 	\SetKwFunction{Union}{Union}\SetKwFunction{FindCompress}{FindCompress}
 	\SetKwInOut{Input}{input}\SetKwInOut{Output}{output}
 	\Input{Convergence criteria $0<\upsilon<<1$, maximum number of iterations $K$.}
    At each DC $d^i\in\mathcal{D}$, choose an arbitrary initial value for $\lambda^i(1),\gamma^i(1),\Lambda^i(1)$.\\
    \For{$k=1$ to $K$}{
    At each DC $d^i\in\mathcal{D}$, derive the values of $\lambda^i,\gamma'^i,\Lambda^i$ for the next iteration (k+1) using Eq.~\eqref{eq:full_dist_2}. \label{lst:line:grad}\\
    At each DC $d^i\in\mathcal{D}$, update the value of $\gamma^i$ using Eq.~\eqref{eq:full_dist_3}.\\
    \If{$|\gamma^i(k+1)-\gamma^i(k)|\leq \upsilon$ and $|\Lambda^i(k+1)-\Lambda^i(k)|\leq \upsilon$ and $|\lambda^i(k+1)-\lambda^i(k)|\leq \upsilon$ and $|\gamma^i(k+1)-\gamma^j(k)|\leq \upsilon, \;1\leq  i\neq j \leq n_d$}{
    Go to line~\ref{lst:line:con}.
    }
    }
    Derive the convex relaxed solution described in Eq.~\eqref{eq:ms}.\label{lst:line:con}\\
    Derive the allocation using Eq.~\eqref{eq:weightedmeanSquare}.
    }
 \end{algorithm}
\vspace{-3mm}
\section{Graph Job Allocation in Large-Scale {\GC}s: Decentralized Approach using Cloud Crawling and PAs' Computing Resources}\label{sec:Large}
\noindent Large-scale {\GC}s consist of an enormous number of PAs and DCs. This fact imposes three challenges for graph job allocation: i) The CDGA algorithm developed above  becomes infeasible. In particular, excessive computational burden will be incurred on the DCs due to the large number of arriving jobs. Also, CDGA in large-scale {\GC}s will incur a long delay (e.g., a GDCN with $100$ DCs involves $300$ Lagrangian multipliers and requires hundreds of iterations for convergence), which may render the final solution less effective for the current state of the network. Moreover, continuous communication between the DCs imposes a considerable congestion over the communication links. ii) So far, the inherent assumption in our study is a known set of feasible allocations for the graph jobs. This requires solving the NP-complete problem of sub-graph isomorphism between the graph jobs and the large-scale aggregated network graph, which may take a long time. iii) Even for a given graph job, the size of the feasible allocation set becomes prohibitively large in a large-scale network. For instance, in a fully-connected network of $100$ DCs, each with $10$ slots, the number of feasible allocations for a simple triangle graph job is ${1000\choose 3}\sim 166\times 10^6$. These concerns motivate us to develop \textit{cloud crawlers}, based on which we address the mentioned challenges through a decentralized framework. Here, we use the term ``crawler" to describe the movement between adjacent DCs. This may bear a resemblance to the term \textit{web crawler}. Nevertheless, the cloud crawlers introduced here are fundamentally different from conventional web crawlers (e.g., \cite{ref:crawler1,ref:crawler2,ref:crawler3}). Our cloud crawlers aim to extract suitable sub-graphs from {\GC}s for specified graph job structures when traversing the network, while web crawlers are mainly developed to extract information from Internet URLs by looking for keywords and related documents. 
 \begin{algorithm}[htb!]
 \footnotesize{
 	\caption{Cloud crawling}\label{alg:CloudCrawler}
 	\SetKwFunction{Union}{Union}\SetKwFunction{FindCompress}{FindCompress}
 	\SetKwInOut{Input}{input}\SetKwInOut{Output}{output}
 	\Input{Initial server $d^i\in \mathcal{D}$, $Gjob_j$, the center node $v_c$ and its maximum shortest distance $D$ to the nodes of the graph, size of the suggested strategies $|\mathcal{SA}_j|$.} 
    Initialize a BST ($BST$), $IA$ as a list of list of lists,  $\mathcal{VISITED}=\{\}$, and vector $Feas\_A$ with length $D+1$.\\
    $\mathcal{VISITED}= \mathcal{VISITED} \cup d^i$\\
    $Feas\_A[1]=1$\\
    \For{$r=1$ to $D$}{\label{algLine:1}
    $Feas\_A[r+1]=Feas\_A[r]+|N^{r}_{v_c}|$\\
    } \label{algLine:2}
    Observer the current pdf of the load of the server $f_{\tilde{L}^i}$ \label{lst:line:crawl}\\
     Initialize $IA\_temp$ as a list of list of lists.\\
    \%Completing the incomplete allocations using the slots of current DC:\\
    \For{$r=1$ to $len(IA)$}{\label{algLine:thisDC3}
    $Last\_Alloc= IA[r,len(IA[r])]$\%Obtain the last allocation done for each incomplete allocation. This is a list of $4$ elements (see line \ref{alg:riid}) \\
    $AN=|\mathcal{V}_j|-Last\_Alloc[4]$ \% \#assigned nodes of the job \\
    $j=find(Feas\_A==AN)+1$ \%Next neighborhood that needs to be assigned\\
    $SA=0$\\
	\While{$j\leq D+1$}{
    $SA=SA+|N^{j}_{v_c}| $\%\#used slots from the current DC\\
    \If{$SA\leq |\mathcal{S}^i|$}{
    $p=E\left\{\tilde{\pi}^i\left(\tilde{L}^i+SA\right)\right\}$\label{algPower}\\
    $LL\_temp=IA[r]$ \%Initialize a temporary list\\
    $LL\_temp.append([d^i,SA,p,Last\_Alloc[4]-SA])$\\
    \If{$j = D+1$}{
    \%Add completed allocations to the BST\\
    \hspace{-3mm}$Tot\_p = Find\_Tot\_Cost(LL\_temp)$ \%Algorithm~\ref{alg:power_tot}\\
    \hspace{-3mm}$Alloc=Create\_Alloc(LL\_temp)$\%Algorithm~\ref{alg:cAlloc}\\
    \hspace{-3mm}$BST=BST\_Add(BST,\underbrace{Tot\_p}_{key},\underbrace{Alloc}_{value},|\mathcal{SA}_j|)$\\
    \vspace{-3mm}
    \%Algorithm~\ref{alg:BSTadd}
    }\Else{
    \hspace{-3mm}$IA\_temp.append(LL\_temp)$
    }
    }
        $j=j+1$
    }
    }
    $IA = IA\_temp$\\ \label{algLine:thisDC4}
    \%Assigning the nodes to the current DC:\\
    \For{$r\;\; in\;\; Feas\_A$}{ \label{algLine:thisDC1}
    \If{$r \leq |\mathcal{S}^i|$}{
    $p=E\left\{\tilde{\pi}^i\left(\tilde{L}^i+r\right)\right\}$\\
    $RS=|\mathcal{V}_j|-r$ \%Number of unassigned nodes of the job\\
    \If{$RS=0$}{
    \hspace{-3.8mm}\%The allocation corresponding to assigning all the nodes to the current server is added to the BST\\
    $BST\_Add(BST,p,[d^i,r],|\mathcal{SA}_j|)$ \%Algorithm~\ref{alg:BSTadd}
    }\Else{
    \hspace{-3.8mm}\%A new incomplete allocation is added to $IA$ as a list of list\\
     $IA.append([[d^i,r,p,RS]])$\label{alg:riid}
    }
    }
    }\label{algLine:thisDC2}
    \If{All the adjacent DC are in the set $\mathcal{VISITED}$}{\label{algLine:nextDC}
    \hspace{-3mm} Initialize a new $IA$ and randomly choose one adjacent DC~$d^k$\\
     \hspace{-3mm} $\mathcal{VISITED}=\{\}$
    }\Else{
    Randomly choose one adjacent DC $d^k$\\
    }
    $\mathcal{VISITED}=\mathcal{VISITED} \cup \{d^k\}$\\
    $i=k$\\
    crawl to $d^i$ and 
    go to line \ref{lst:line:crawl}\label{algLine:nextDC2} \\
       }
 \end{algorithm}

\subsection{Strategy Suggestion Using Cloud Crawling}
We introduce a cloud crawler (CCR) which carries a collection of structured information traveling between adjacent DCs. It probes the connectivity among the DCs and status of them (power usage, load distribution, etc.), based on which it provides a set of suggested allocations for the graph jobs. For a faster network coverage, multiple CCRs for each type of graph job can be assumed. Information gleaned by the CCRs can be shared with the PAs who act as mediators between the {\GC} and customers using two mechanisms: i) the CCR shares them with a central database, which PAs have access to, on a regular basis; ii) the CCR shares them with DCs as it passes through them and the DCs update the connected PAs accordingly. 
The goal of a CCR is to find ``potentially good" feasible allocations to fulfill a graph job's requirements considering the network status. We consider a potentially good feasible allocation as a sub-graph in the aggregated network graph which is isomorphic to the considered graph job leading to a low cost of execution. In the following, we first prove a theorem, based on which we provide a corollary aiming to describe a fast decentralized approach to solve the sub-graph isomorphism problem in large-scale {\GC}s.  
 
\begin{definition}\label{def:1}
Two graphs $G$ and $G'$ with vertex sets $\mathcal{V}$ and $\mathcal{V}'$ are called isomorphic if there exists an isomorphism (bijection mapping) $g: V \rightarrow V'$ such that any two nodes, $\Scale[.97]{a,b\in \mathcal{V}}$, are adjacent in $G$ if and only if $\Scale[.97]{g(a),g(b)\in \mathcal{V'}}$ are adjacent in $G'$.
\end{definition}
  \begin{algorithm}[h]
 {\footnotesize
 	\caption{$Find\_Tot\_Cost$}\label{alg:power_tot}
 	\SetKwFunction{Union}{Union}\SetKwFunction{FindCompress}{FindCompress}
 	\SetKwInOut{Input}{input}\SetKwInOut{Output}{output}
    \Input{A list $LL$}
       \Output{The total cost $C$.}
    $C=0$\\
    $j=0$\\
    \While{LL[j]!=null}{
    \hspace{-3mm}$C=C+LL[j][3]$ \% Sum the incurred costs on all the DCs involved\\
    \hspace{-3mm}$j=j+1$
    }
 return C
 }
 \end{algorithm}
  \begin{algorithm}[h]
 {\footnotesize
 	\caption{$Create\_Alloc$}\label{alg:cAlloc}
 	\SetKwFunction{Union}{Union}\SetKwFunction{FindCompress}{FindCompress}
 	\SetKwInOut{Input}{input}\SetKwInOut{Output}{output}
    \Input{A list $LL$}
    \Output{Allocation strategy $S$}
    Initialize $S$ as a list of lists\\
    $j=0$\\
    \While{LL[j]!=null}{
    \hspace{-3mm}$S=S.append([LL[j][1],LL[j][2]])$ \%The DC's index and its number of used slots\\
    \hspace{-3mm}$j=j+1$\\
    }
 return S
 }
 \end{algorithm}
\begin{algorithm}[h]
 {\footnotesize
 	\caption{$BST\_Add$}\label{alg:BSTadd}
 	\SetKwFunction{Union}{Union}\SetKwFunction{FindCompress}{FindCompress}
 	\SetKwInOut{Input}{input}\SetKwInOut{Output}{output}
    \Input{A binary search tree $BST$, a $key$ and a $value$, the desired size of the suggested strategy set $|\mathcal{SA}_j|$}
    \Output{A binary search tree $BST$}
   \If{BST.length $< |\mathcal{SA}_j|$}{
   BST=BST.Insert($key$,$value$)
   }\ElseIf{$key < $ BST.get\_max().key }{
    BST=BST.Delete(BST.get\_max())\\
    BST=BST.Insert($key$,$value$)
   } 
 return BST
 }
 \end{algorithm}

\begin{theorem}\label{th:graphIso}
Consider graphs $G$ and $H$ with vertex sets $\mathcal{V}_G$ and $\mathcal{V}_H$, respectively, where $|\mathcal{V}_G| \leq |\mathcal{V}_H|$. Assume that $H$ can be partitioned into multiple complete sub-graphs $h_0,...,h_N$, $N\geq1$, with vertex sets $\mathcal{V}_{h_0},,...,\mathcal{V}_{h_N}$, where $\cup_{i=0}^{N} \mathcal{V}_{h_i}=\mathcal{V}_H$, and all the nodes in each pair of sub-graphs with consecutive indices are connected to each other. Consider node $v \in \mathcal{V}_G$ and let $D\leq N$ denote the length of the longest shortest path between $v$ and nodes in $\mathcal{V}_G$. Define $\mathcal{N}^k_{v} \triangleq \{\hat{v}\in \mathcal{V}_G : SP(\hat{v},v)=k\}$, $\mathcal{N}^0_{v}=\{v\}$, where $SP(.,.)$ denotes the length of the shortest path between the two input nodes.
Let $\{i_j\}_{j=0}^{D}$ be a sequence of integer numbers that satisfy the following conditions:
\begin{align}
&\hspace{-1mm}0\leq i_0,i_1,i_2,\cdots,i_{D} \leq D+1,\label{con:1}\\
&\hspace{-1mm}\sum_{l=0}^{D} i_l=D+1, \\
&\hspace{-1mm}i_{j}=0 \Rightarrow i_{j+1}=0, \;\;\forall j\in\{0,\cdots,D-1 \}.\label{con:3}
\end{align}
For such a sequence $\{i_j\}_{j=0}^{D}$, there is at least an isomorphic sub-graph to $G$, called $G'$, in $H$ with the corresponding isomorphism mapping $g$, for which at least one of the nodes of $G'$, $v'=g(v)$, belongs to $h_0$, if the following set of conditions is satisfied:
\vspace{-1mm}
\begin{equation}\label{eq:proover}
\begin{aligned}
\begin{cases}
|\mathcal{V}_{h_0}|\geq \sum_{i=0}^{i_0-1}|\mathcal{N}^{i}_{{v}}|,\\
|\mathcal{V}_{h_1}|\geq \sum_{i=1}^{i_1} \mathbf{1}_{\{i_1 \geq 1\}}|\mathcal{N}^{i+i_0-1}_{{v}}|, \\
|\mathcal{V}_{h_2}|\geq \sum_{i=1}^{i_2} \mathbf{1}_{\{i_2 \geq 1\}}|\mathcal{N}^{i+i_0+i_1-1}_{{v}}|, \\
\vdots\\
|\mathcal{V}_{h_{D}}|\geq \sum_{i=1}^{i_D}\mathbf{1}_{\{i_D \geq 1\}}|\mathcal{N}^{i+\sum_{l=0}^{D-1} i_{l} -1}_{{v}}|.
\end{cases}
\end{aligned}
\end{equation}
\end{theorem}
\begin{proof}
 The key to prove this theorem is considering the following mapping between the nodes of $G$ and the sub-graphs in $H$:
\begin{equation}\label{eq:middproof}
   \hspace{-3.4mm} [v \rightarrow h_0, \;\mathcal{N}^1_{{v}}\rightarrow h_1,\; \mathcal{N}^2_{{v}}\rightarrow h_2, \cdots, \;\mathcal{N}^D_{{v}}\rightarrow h_D].
\end{equation}
Under this mapping, the mapped nodes form an isomorphic graph to $G$ since the connection between all the adjacent nodes in $G$ is met in $H$. That is because they are either placed at the same (fully-connected) $h_i$, $0\leq i \leq N$ or in (fully-connected) adjacent $h_i$-s, $0\leq i \leq N$. With a similar justification, it can be proved that concatenation of the mapped nodes to the adjacent $h_i$-s, $0\leq i \leq N$, in Eq.~\eqref{eq:middproof} preserves the isomorphic property. For instance, all the following mappings form isomorphic graphs to $G$ in $H$:
\begin{align}
        &\hspace{-2mm}[v \rightarrow h_0, \;\mathcal{N}^1_{{v}}\rightarrow h_1,\; \cdots,\mathcal{N}^{D-2}_{{v}}\rightarrow h_{D-2}, \nonumber \\
        &\hspace{-2mm}\;\mathcal{N}^{D-1}_{{v}} \cup \mathcal{N}^D_{{v}}\rightarrow h_{D-1}, \{\}\rightarrow h_{D}],\\
        &\hspace{-2mm}[v \rightarrow h_0, \;\mathcal{N}^1_{{v}}\rightarrow h_1,\; \cdots,\mathcal{N}^{D-3}_{{v}}\rightarrow h_{D-3}, \nonumber \\
        &\hspace{-2mm}\;\mathcal{N}^{D-2}_{{v}} \cup \mathcal{N}^{D-1}_{{v}} \cup \mathcal{N}^{D}_{{v}}\hspace{-1mm}\rightarrow \hspace{-1mm} h_{D-2}, \{\}\hspace{-1mm}\rightarrow h_{D-1}, \{\}\hspace{-1mm}\rightarrow h_{D}],\\
        &\hspace{-2mm}[v \cup \mathcal{N}^1_{{v}} \rightarrow h_0, \;\mathcal{N}^2_{{v}}\rightarrow h_1,\; \cdots,\mathcal{N}^{D-1}_{{v}}\rightarrow h_{D-2}, \nonumber \\
        &\hspace{-2mm}\;\mathcal{N}^{D}_{{v}} \rightarrow h_{D-1}, \{\}\rightarrow h_{D}],\\
        &\hspace{-2mm}[{v} \rightarrow h_0, \;\mathcal{N}^1_{{v}} \cup \mathcal{N}^2_{{v}} \cup \mathcal{N}^3_{{v}}\rightarrow h_1, \nonumber \\
        &\hspace{-2mm}\mathcal{N}^4_{{v}} \rightarrow h_2,  \cdots,\;\mathcal{N}^{D}_{{v}}\rightarrow h_{D-2}, \hspace{-2mm}\;\{\} \rightarrow h_{D-1}, \{\}\rightarrow h_{D}].
\end{align}
It can be seen that conditions stated in Eq.~\eqref{con:1}-\eqref{con:3} denote the feasible concatenation strategies, where each $i_j$ denotes the number of neighborhoods mapped to $h_j$, $0\leq j \leq D$. Also, Eq.~\eqref{eq:proover} ensures the feasibility of the corresponding mappings.  
\end{proof}

\begin{corollary}\label{cor:1}
For $Gjob_j$, assume a CCR located at DC $d^{j_0}$ allocating at least one node of $Gjob_j$, $v_c\in \mathcal{V}_j$, to one slot at $d^{j_0}$, where the length of longest shortest path between $v_c$ and nodes in $\mathcal{V}_j$ is $D$. Assume that the CCR's near future path can be represented as $d^{j_0} \rightarrow d^{j_1} \cdots \rightarrow d^{j_{D}} $, where $j_i\neq j_k,\; \forall i\neq k$. Considering $d^{j_i}$ as $h_i$ in Theorem~\ref{th:graphIso}, for each realization of the sequence $\{i_j\}_{j=0}^{D}$ satisfying Eq.~\eqref{con:1}-\eqref{eq:proover}, the following described allocation $\mathbf{M}_j=[m^1_j,m^2_j,\cdots,m^{n_d}_j]$ is feasible and is isomorphic to $Gjob_j$.
\begin{equation}
\begin{aligned}
m^k_j=\begin{cases}
\sum_{i=0}^{i_0-1}|\mathcal{N}^{i}_{v_{c}}| &k={j_0},\\
\sum_{i=1}^{i_1} \mathbf{1}_{\{i_1 \geq 1\}}|\mathcal{N}^{i+i_0-1}_{v_{c}}|&k=j_1, \\
\sum_{i=1}^{i_2} \mathbf{1}_{\{i_2 \geq 1\}}|\mathcal{N}^{i+i_0+i_1-1}_{v_{c}}|&k=j_2, \\
\vdots\\
\sum_{i=1}^{i_D}\mathbf{1}_{\{i_D \geq 1\}}|\mathcal{N}^{i+\sum_{l=0}^{D-1} i_{l} -1}_{v_{c}}| &k=j_D,\\
0 \;& Otherwise.
\end{cases}
\end{aligned}
\end{equation}
\end{corollary}

Using our method described in the above corollary, it can be verified that the complexity of obtaining an isomorphic sub-graph to a graph job for a CCR becomes $O(D)$, where $D$ is the diameter of the graph job. Henceforth, we recall $v_c$ defined in Corollary~\ref{cor:1} as the center node, which can be chosen arbitrarily from the graph job's nodes. The pseudo-code of our algorithm implemented in a CCR is given in Algorithm~\ref{alg:CloudCrawler}. We use the binary search tree (BST) data structure~\cite{ref:DS} to structurize the carrying suggested strategies. To handle the large number of feasible allocations, we limit the capability of a CCR in carrying potentially good strategies (size of the BST) to a finite number $|\mathcal{SA}_j|$ for $Gjob_j\in\mathcal{J}$. Some important parts of Algorithm~\ref{alg:CloudCrawler} are further illustrated in the following. 

\textbf{A) Initialization:} A CCR is initialized at a DC for a certain graph job, $Gjob_j\in\mathcal{J}$, and a specified number of suggested strategies ($|\mathcal{SA}_j|$) to be carried.\footnote{Note that using a simple extension of this algorithm, a CCR can handle the extraction of suggested strategies for multiple graph jobs at the same time.} Each CCR carries a BST, a list~\cite{list} of incomplete allocations ($IA$) and a set of visited neighbors ($Visited$) (can be implemented as a list). In Fig.~\ref{fig:CCR}-a, topology of a graph job is shown along with three DCs where each square denotes a slot in a DC. The CCR is initialized at $d^1$ traversing the path $d^1\rightarrow d^2\rightarrow d^3$.

\textbf{B) Determining the Graph Job Topology Constraints (lines:~\ref{algLine:1}-\ref{algLine:2}):} For a given center node of a graph job, i.e., $v_c$ in Corollary~\ref{cor:1}, the algorithm calculates the feasible number of nodes allocated to DCs according to Corollary~\ref{cor:1}. In Fig.~\ref{fig:CCR}-a, the center node is denoted by $v_c$ and different set of neighbors located in various shortest paths to $v_c$ are demonstrated.  

\textbf{C) Allocation Initialization and Completion (lines: \ref{algLine:thisDC3}-\ref{algLine:thisDC2}):} 
\begin{figure*}[t]
	\centering
	\includegraphics[width=1\textwidth]{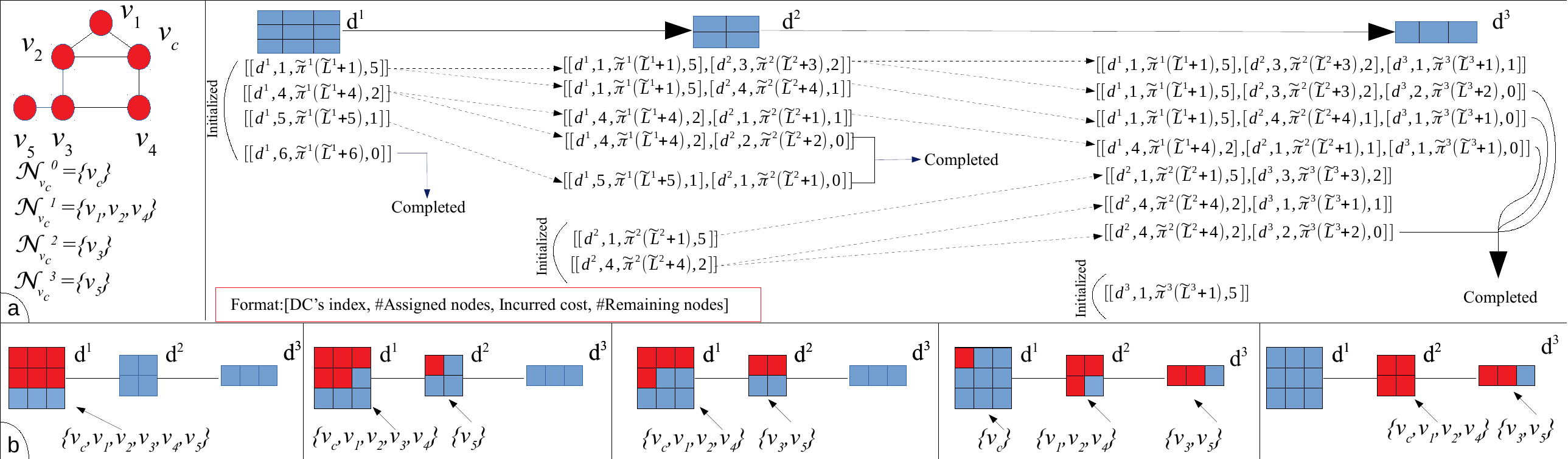}
	\caption{a: The graph job topology and the neighboring nodes to the center node (left); three DCs along with the carried incomplete and complete allocations of the CCR upon arriving at each DC (right). b: Some examples of completed allocations.}
	\label{fig:CCR}
	\vspace{-4mm}
\end{figure*}
 According to Corollary~\ref{cor:1}, the crawler attempts to complete the incomplete allocations in $IA$ by accommodating the remaining nodes of the graph job to the current DC (lines: \ref{algLine:thisDC3}-\ref{algLine:thisDC4}). During this process, if the remaining number of unassigned nodes of the graph job becomes zero, the corresponding allocation is added to the BST considering its incurred cost. The rest of the allocations are added to the $IA$. Also, the allocations are initialized upon arriving at each DC using Corollary~\ref{cor:1} (lines: \ref{algLine:thisDC1}-\ref{algLine:thisDC2}). In Fig.~\ref{fig:CCR}-a, the initialized allocations are depicted underneath $d^1$. Also, the updated set of incomplete allocations and completed allocations during the movement of the CCR are depicted underneath $d^2$ and $d^3$. Also, some of the completed allocations are depicted in Fig.~\ref{fig:CCR}-b for a better understanding.

\textbf{D) Traversing the {\GC} (lines: \ref{algLine:nextDC}-\ref{algLine:nextDC2}):} The CCR examines the adjacent DCs of its current location. If there are multiple non-visited neighbor DCs, the CCR chooses its next destination randomly among them. However, if all the neighbor DCs are visited, the CCR clears the $visited$ set, and chooses its next destination at random. This process is designed to avoid re-visiting the previously visited DCs or trapping at a DC in which all of its neighbor DCs are visited.

\begin{remark}
After the size of the BST reaches the predefined length ($|\mathcal{SA}_j|$), a new completed strategy is added to the BST if it possesses a lower incurred cost as compared to the strategy with the maximum incurred cost in the BST, and the latter is deleted subsequently.
\end{remark}
\begin{remark}
It is assumed that each DC has a probabilistic prediction for its near future load distribution. Hence, the crawler obtains the expected cost of allocation in line \ref{algPower}, e.g., $E\{\tilde{\pi}^i(\tilde{L}^i+m)\}= \sum_{j=0}^{|\mathcal{S}^i|-m} \pi(j+m) f_{\tilde{L}^i}(\tilde{L}^i=j)$ when $m$ slots of DC $d^i$ are taken, where $f_{\tilde{L}^i}$ is the probability mass function of the predicted load of $d^i$ and $\pi(j+m)$ is the incurred cost stated in Eq.~\eqref{eq:powerPrice} with $L^i+m^i_j$ replaced with $j+m$.
\end{remark}
\begin{remark}
In the BST considered (see Algorithm~\ref{alg:CloudCrawler}), each node has two attributes: $``key"$ and $``value"$, where $key$ is a real number and $value$ is a list. The functions $getmax()$, $Delete()$, and $Insert()$ are assumed to be known, for which sample implementation can be found in~\cite{ref:DS}. Also in Algorithm~\ref{alg:CloudCrawler}, the function $len()$ returns the length of the input argument. If the input is a list, it returns the number of elements; if the input is a list of lists, it returns the number of lists inside the outer-list, etc. Moreover, in Algorithm~\ref{alg:BSTadd}, the $``length"$ attribute indicates the number of nodes of the BST. 
\end{remark}
\begin{remark}
The BST is used for its unique characteristics. If the BST is balanced, (e.g., implemented as an AVL tree) this data structure enables deletion of the strategy with the maximum cost and insertion of a new strategy, which are both necessary in the CCR, in time complexity of $O(\log |\mathcal{SA}_j|)$. Moreover, a simple \textit{inorder traversal}, which can be done in $O(|\mathcal{SA}_j|)$, gives the suggested strategies in ascending order with respect to their incurred cost. 
\end{remark}
\begin{remark}
We designed CCRs mainly for allocation of graph jobs in large-scale GDCNs. However, they could also be implemented in small- and medium-scale GDCNs. In those cases, if the CCR continuously explores the network and the power consumption of DCs changes smoothly over time, the allocation cost of the best strategy in the BST of the CCR would be similar to those of the solutions obtained in Section~\ref{sec:Small} and ~\ref{sec:Medium}. Note that,
the analytical solutions proposed for small- and medium-scale GDCNs are guaranteed to find the optimal solution of $\mathcal{P}3$, which, due to the limitations discussed
at the beginning of this section, are not applicable in large-scale GDCNs.
However, cloud crawling can be viewed as an empirical approach, which not only
provides a distributed solution, but also offers additional benefit such as extraction of isomorphic sub-graphs to a given graph job.
\end{remark}
  Fig.~\ref{diag:sysmodel} depicts a sample CCR traversing over the network, where its corresponding information is shown in the bottom left of the figure. In this figure, a crawler is considered attempting to assign a graph job with $7$ nodes to the network. It is assumed that given the center node $v_c$, we have: $\mathcal{N}^0_{v_{c}}=1$, $\mathcal{N}^1_{v_{c}}=2$, $\mathcal{N}^2_{v_{c}}=2$, $\mathcal{N}^3_{v_{c}}=1$, and $\mathcal{N}^4 _{v_{c}}=1$. In the depicted BST, each suggested strategy is a list of lists, each of which consists of two elements: index of a DC and the number of slots utilized from that DC. 
  
  So far, PAs are provided with a pool of potentially good allocations using CCRs. In the following, we address suitable strategy selection approaches for PAs with respect to the pricing policy of the DCPs. 

\subsection{Strategy Selection Under Fixed Pricing}\label{sec:regret}

 Due to the simplicity of implementation, fixed pricing is still a common approach to offer cloud services to customers. In this case, DCPs determine a constant price for utilizing each slot of their DC, which is chosen with respect to the expected load of the DC to guarantee a certain amount of profit. In this subsection and Subsection~\ref{sec:correquib}, it is assumed that PAs assign their graph jobs to the system in a sequential manner, where at each \textit{iteration} each PA assigns (at most) one graph job of each type (if it is requested by a customer) to the system.
\subsubsection{Problem Formulation}
We formulate the problem from the perspective of one PA since the utilization cost of DCs are assumed to be constant. For the $n^{th}$ arriving $Gjob_j$, the PA chooses an allocation $\mathbf{M}_{j,(n)}=[m^1_{j,(n)},\cdots,{m}^{n_d}_{j,(n)}]$ from the pool of the CCR's suggested allocations $\mathcal{SA}_j$. In this case, we define the utility function of a PA as:
\begin{align}\label{eq:util2}
&U_j(n)|_{\mathbf{M}_{j,(n)}}  \triangleq \rho_j -\chi_j \sum_{k=1}^{n_d} PC^k {m}^k_{j,(n)}
\nonumber \\ 
&-\phi_j  \sum_{k=1}^{n_d} \mathbf{1}_{\{{m}^k_{j,(n)}>0\}}+ \chi_j |\mathcal{V}_j|PC^{max}+\phi_j |\mathcal{V}_j|,
\end{align}
where $PC^.$ denotes the slot cost of the indexed DC and $PC^{max}$ is a constant larger than the price of all the slots in the system. In this expression, different preferences of PAs are governed by positive real constants $\phi_j$ and $\chi_j$, $\rho_j \in \mathbb{R}^+$ is the default reward of execution, the second term describes the payment, the third term describes the privacy preference of the PA, and the last two terms are added to ensure that the utility function is non-negative. In the third term, a large value of $\phi_j$ implies more tendency toward utilizing fewer DCs to execute the graph job. The normalized utility function can be derived as: $\tilde{U}_j(n)= U_j(n)/\big[\rho_j+\chi_j|\mathcal{V}_j|PC^{max} + \phi_j (|\mathcal{V}_j|-1)\big]$. 

In this context, each PA aims to maximize his utility by selecting the best sequence of allocations ${\widetilde{M}_j}^*$. Mathematically:
\begin{align}
&{\widetilde{M}_j}^*= \argmax_{\widetilde{M}_j=\{\mathbf{M}_{j,(n)}\}_{n=1}^{N_j}} \;\;\sum_{n=1}^{N_j} U_j(n) |_{\mathbf{M}_{j,(n)}} \nonumber \\
&{\textrm{s.t.} } \;\;\mathbf{M}_{j,(n)} \in \mathcal{SA}_j,\; \;\forall n\in \{1,\cdots,N_j\}.
\end{align}
 Due to accessibility constraints, a newly joined PA may not have complete information about the prices of all the slots. Also, DCPs may update the utilization costs periodically. Hence, initially there is a lack of knowledge about the DCs' prices on the PAs' side making conventional optimization techniques inapplicable. We tackle this problem by proposing an online learning algorithm partly inspired by the concept of \textit{regret}. This concept originates in the multi-armed bandit problem~\cite{ref:gambler}, where the
gambler aims to identify the best slot machine to play (best
strategy) at each round of his gambling while considering the
history of the rewards of the machines. Our algorithm is an advanced version of the original algorithm in~\cite{ref:gambler} tailored for the graph job allocation in {\GC}s. 

\subsubsection{Boosted Regret Minimization Assignment (BRMA) Algorithm}
By choosing a strategy from the set of suggested strategies of a CCR and observing the utility, one can get an estimate about the utility of the similar strategies targeting similar DCs. To address this, in our algorithm, we use the concept of \textit{k-means clustering}~\cite{ref:clustering} to partition the strategies into different groups with respect to their similarity. Let $\mathcal{C}=\{C_1,C_2,\cdots,C_{|\mathcal{C}|}\}$ denote the set of clusters obtained using the method of~\cite{ref:clustering}. We group consecutive iterations of our algorithm as a ``time-frame", according to which $\mathcal{T}=\{tf_1,tf_2,\cdots, tf_{\ceil{N/\Gamma}}\}$ denotes the set of time-frames, where $N$ is the number of iterations and $\Gamma$ is the time-frame length. In this case, iterations $1$ to $\Gamma$ belong to $tf_1$, iterations $\Gamma+1$ to $2\Gamma$ belong to $tf_2$, etc. Let $\mathcal{A}_{tf_k}$ denote the set of actions performed in the $k^{th}$ time-frame and $A_{k}$ denote the action executed at iteration $k$. For strategy $m\in \mathcal{SA}_j$, let $ \kappa^k_m=\sum_{n=(k-1)\Gamma+1}^{k\Gamma} \mathbf{1}_{\{A_{n}=m\}} $ denote the number of times the strategy is chosen during $tf_k$. 
The pseudo-code of our proposed algorithm is given in Algorithm~\ref{alg:regretmin}. The main differences between our proposed algorithm and the method in~\cite{ref:gambler} are as follows: i)  the concept of clustering is leveraged to group the analogous strategies; ii) a new weight update mechanism is proposed based on the concept of ``similarity", with which the weights of the unutilized strategies are estimated employing the utility of the chosen strategies; iii) the concept of time-frame is incorporated in our design (see Remark~\ref{remark:Grouping}). These approaches significantly improve the speed of convergence of the algorithm (see Section~\ref{sec:Simu}). In our algorithm, during $tf_n$, every time a PA needs to allocate $Gjob_j$, he chooses a strategy ($m\in \mathcal{SA}_j$) with probability:
\begin{align} \label{eq:prob}
\mathbf{p}_{j,(m)}(tf_n)&= (1-E) \frac{w_{j,(m)}(tf_n)}{\sum_{a \in \mathcal{SA}_j} w_{j,(a)}(tf_n)} + \frac{E}{|\mathcal{SA}_j|},
\end{align}
where $w_{j,(m)}(tf_n)$ denotes the current weight of strategy $m$ and $0<E<1$. The second term is introduced to avoid trapping in local maxima.  

\textbf{Description of the BRMA Algorithm:} Initially, all the strategies have the same weight and the same probability of selection. At each iteration, one strategy is chosen according to the probability of selection. At the end of each time-frame, virtual rewards of the chosen strategies are derived using Eq.~\eqref{eq:Q}. Also, the utility of those strategies in a cluster from which one member is chosen is estimated using Eq.~\eqref{eq:Q2}. To obtain the estimation, we propose the following similarity index: $\frac{\exp \left(\frac{k_s<A_z,m>}{||A_z||\; ||m||}\right)-1}{\exp({k_s})-1}$, where $k_s>>1$ is a real constant. This index is maximized when two strategies utilize the same exact DCs with the same number of slots, and it is zero when they have no used DCs in common. The weights for the next time-frame are obtained using Eq.~\eqref{eq:weight} ($K$ is a positive real constant) followed by obtaining the probability of selections.
\begin{algorithm}[t]
 	\caption{BRMA: Boosted~regret~minimization assignment}\label{alg:regretmin}
 	 {\footnotesize
 	\SetKwFunction{Union}{Union}\SetKwFunction{FindCompress}{FindCompress}
 	\SetKwInOut{Input}{input}\SetKwInOut{Output}{output}
    \Input{Length of time-frames $\Gamma$, $\mathcal{SA}_j$, number of time frames $TN_j$.}
	Obtain clusters $\mathcal{C}=\{C_1,C_2,\cdots,C_{|\mathcal{C}|}\}$ from $\mathcal{SA}_j$ according to~\cite{ref:clustering}.\\
    Assign $w_{j_i}(tf_1)=1$, $\mathbf{p}_{j_i}(tf_1)=1/|\mathcal{SA}_j|$, $\forall i \in \mathcal{SA}_j$.\\
    \For{$n=1$ to $TN_j$}{
   Choose a strategy $m\in \mathcal{SA}_j$ according to $\mathbf{p}_j(tf_{\ceil{n/\Gamma}})$ ($A_n=m$) and observe the normalized utility $\tilde{U}_j(n)|_{m}$.\\
    \If{$n=\Gamma k$, $k \in \mathbb{Z}^+$}{
    Obtain the virtual reward for each $m \in \mathcal{A}_{tf_k}$ as follows:
    	\vspace{-2.5mm}
    \begin{equation}\label{eq:Q}
\hspace{-7mm}Q'_{j,(m)}(tf_k)\hspace{-.5mm}=\hspace{-.5mm} \frac{\displaystyle\sum_{z=(k-1)\Gamma+1}^{k \Gamma} \tilde{U}_j(z)|_{A_z} \mathbf{1}_{\{A_z=m\}}}{\kappa^k_m}.
\end{equation}\\
Obtain the virtual reward for those strategies $m \notin \mathcal{A}_{tf_k}$ which at least one strategy from their cluster is chosen as follows:
	\vspace{-2.5mm}
    \begin{equation}\label{eq:Q2}
  \begin{aligned}
Q'_{j,(m)}(tf_k)= \frac{1}{|C_{q}| } \Bigg[ \sum_{z={(k-1)\Gamma+1}}^{{k\Gamma}} \mathbf{1}_{\{A_z\in C_{q}\}}\\
\underbrace{\frac{\exp \left(\frac{k_s <A_z,m>}{||A_z||\; ||m||}\right)-1}{\exp({k_s})-1}}_\text{similarity index} \tilde{U}_j(z)|_{A_z} \Bigg] \hspace{0.6mm} 
\textrm{ if } m\in C_{q}.
\end{aligned}
\end{equation}\\
For strategy $m\in\mathcal{SA}_j$, $Q'_{j,(m)}(tf_k)=0$ if neither itself nor any strategy from its cluster is chosen in this time-frame.\\
	Update the weights of the strategies for the next time-frame:
	\vspace{-2.5mm}
    \begin{equation}\label{eq:weight}
\hspace{-7mm}w_{j,(m)}(tf_{k+1})= w_{j,(m)}(tf_{k}) \exp\hspace{-.5mm} \left(\frac{K  Q'_{j,(m)} (tf_{k})}{|\mathcal{SA}_j|}\right),
\end{equation}\\
 Derive the distribution $\mathbf{p}_j(tf_{k+1})$ according to Eq.~\eqref{eq:prob}.\\
    }
    }
    }
 \end{algorithm}
\begin{remark}\label{remark:Grouping}
In case of updating the weights of the strategies at each time instant, selecting multiple strategies with low utilities in the initial iterations boosts the weights of those strategies leading to an undesired low probability of selection for not chosen high utility strategies. To avoid this, we use the concept of time-frame and update the weights at the end of time-frames. 
\end{remark}
\begin{figure*}[htb]
	\centering
	\includegraphics[width=7.1in, height=.8in]{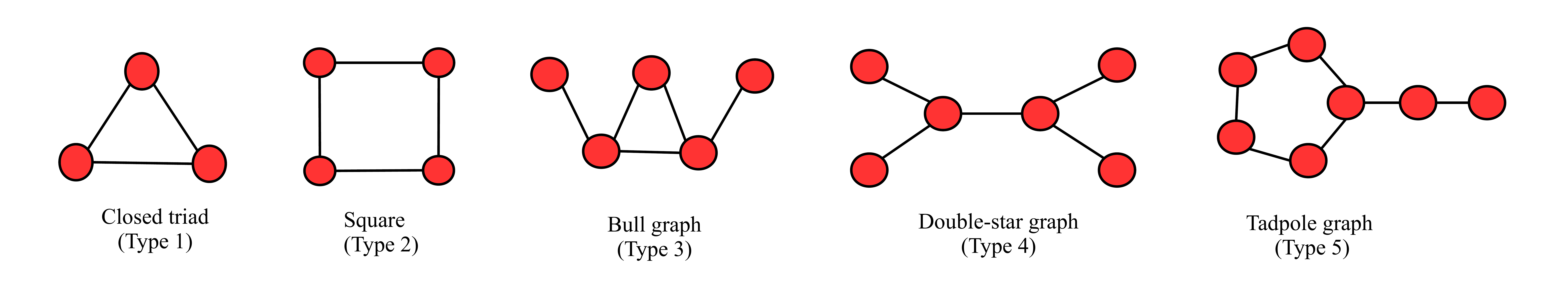}
	\caption{Topology of the graph jobs.}
	\label{fig:graphJobs}
	\vspace{-3mm}
\end{figure*}
\begin{figure*}[htb!]
		\centering
		\subcaptionbox{The incurred power of allocation using the optimal allocation (exhaustive search) and the sub-optimal approach and greedy algorithms.\label{diag:scenario1a}}[.3\linewidth][c]{%
		\includegraphics[width=.330\linewidth]{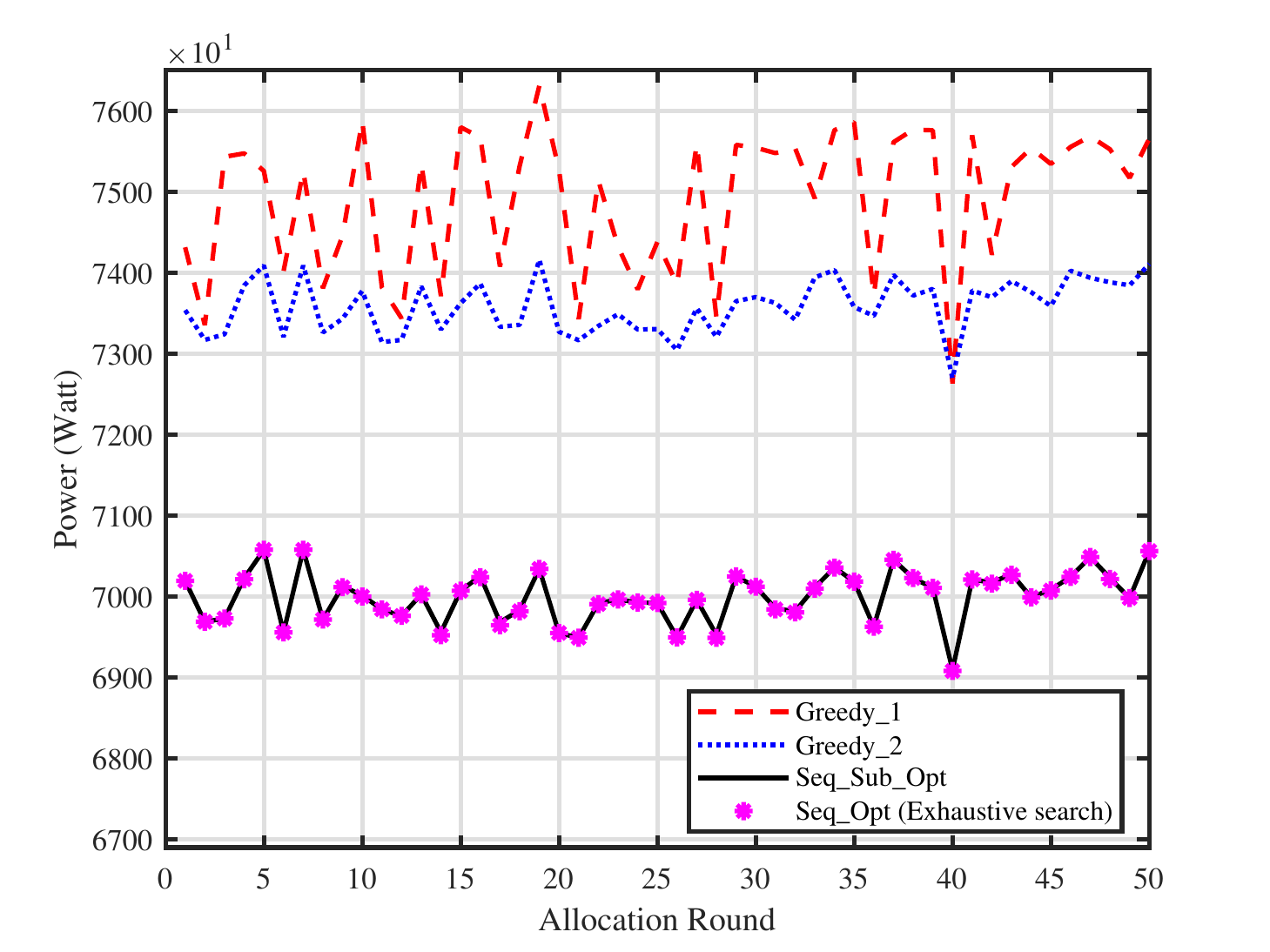}}\quad
		\subcaptionbox{Cumulative average incurred power of allocation using the optimal allocation (exhaustive search) and the sub-optimal approach and greedy algorithms.\label{diag:scenario1b}}[.3\linewidth][c]{%
		\includegraphics[width=.330\linewidth]{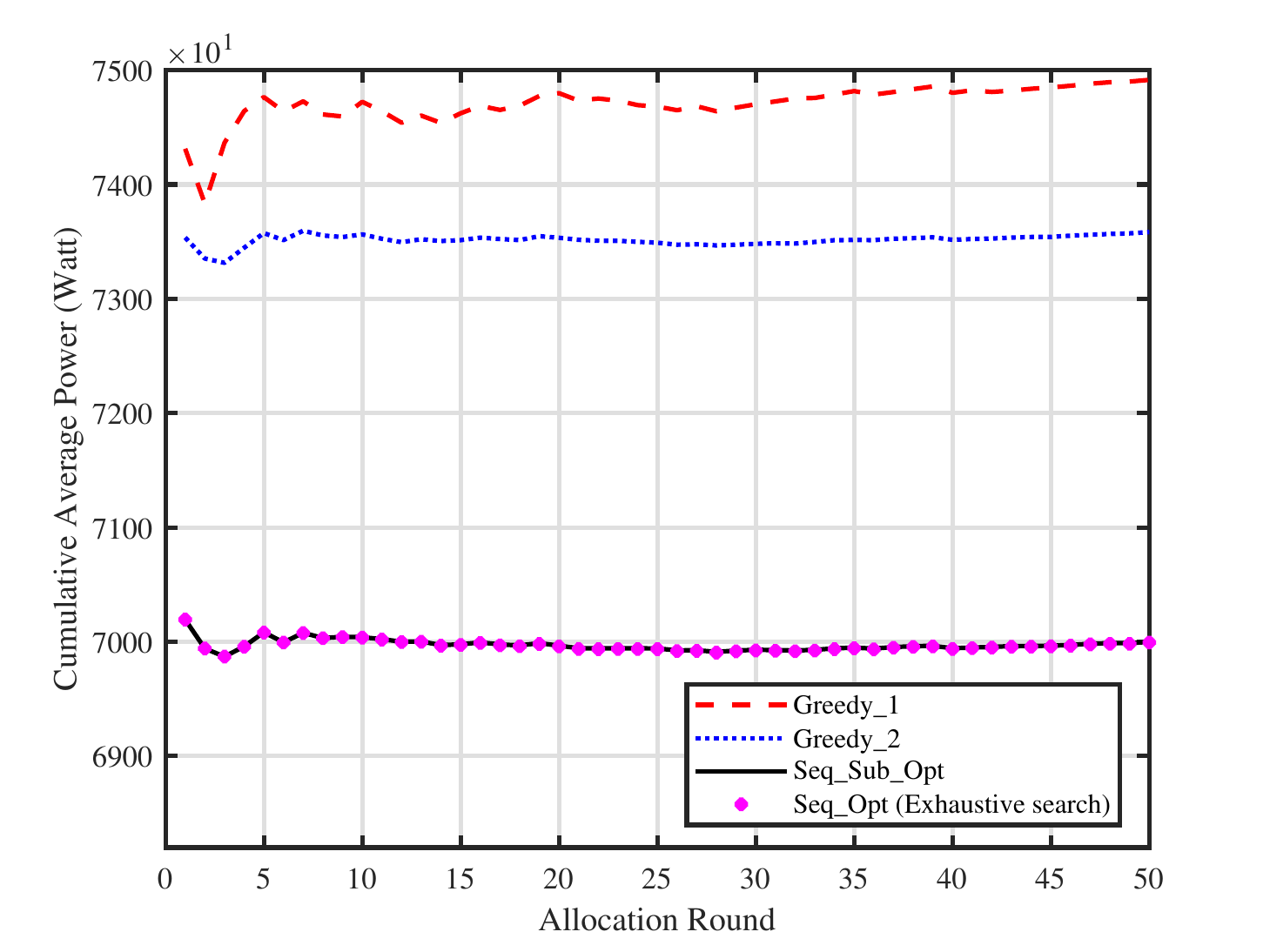}}\quad
		\subcaptionbox{The generated revenue upon using the sub-optimal approach as compared to the greedy methods.\label{diag:scenario1c}}[.3\linewidth][c]{%
		\includegraphics[width=.330\linewidth]{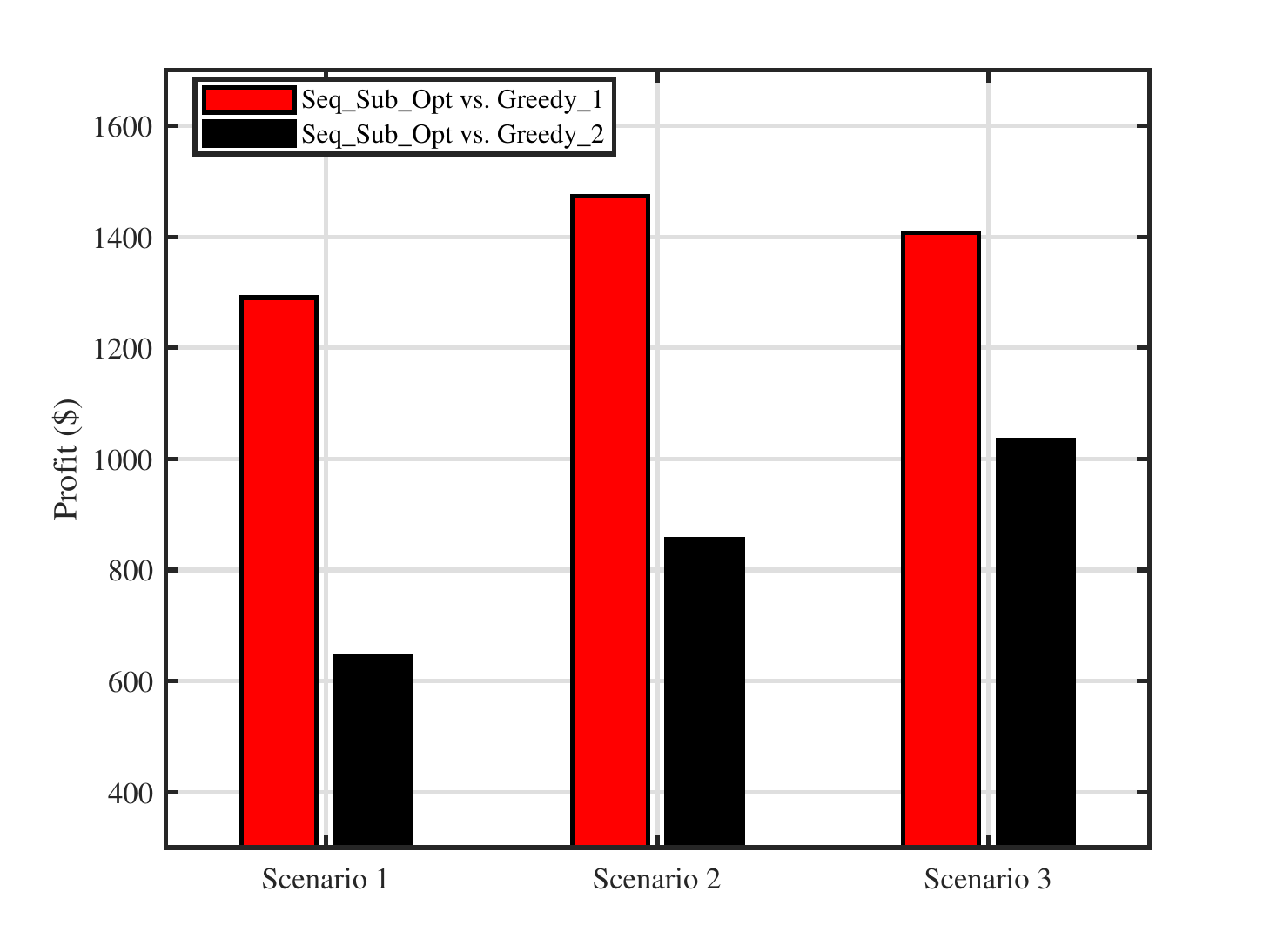}}
       
		\caption{Simulation results of the small-scale {{\GC}}.}
         \label{diag:scenario1}
         \vspace{-5mm}
	\end{figure*}
\subsection{Strategy Selection Under Adaptive Pricing}\label{sec:correquib}
 In modern cloud networks, cloud users are charged in a real-time manner with respect to their incurred load on DCs and the status of the DCs~\cite{zhao2014dynamic, wan2016reactive}. In this work, we propose an adaptive pricing framework suitable for graph job allocation in {\GC}s. Let $\mathcal{P}=\{p_1,\cdots,p_{|\mathcal{P}|}\}$ denote the set of active PAs. For $Gjob_j$, based on Eq.~\eqref{eq:powerPrice}, upon utilizing suggested strategy $\mathbf{M}_{j,(a_k)}\in \mathcal{SA}_j$, we model the total payment of PA $p_k\in \mathcal{P}$ to the DCPs as:
\begin{equation}\label{eq:loadprice}
\begin{aligned}
&\Upsilon(\mathbf{M}_{j,(a_k)},\{\mathbf{M}_{j,({a_{k'}})}\}_{k'=1, k'\neq k}^{|\mathcal{P}|})=\sum_{i=1}^{n_d} m^i_{j,(a_k)}\xi^i \nu^i+\\
&\sum_{i=1}^{n_d}\frac{   \xi^i \eta^i N^i  \left(\hspace{-1mm} \sigma^i \hspace{-1mm} \left(\frac{L^{i}+ \sum_{k=1}^{|\mathcal{P}|}m^i_{j,(a_k)}}{|\mathcal{S}^i|}\right)^{\alpha^i}\hspace{-3mm}+P^i_{idle}  \right)}{\sum_{k=1}^{|\mathcal{P}|} m^i_{j,(a_k)}}m^i_{j,(a_k)}.
\end{aligned}
\end{equation}
Consequently, we model the utility of PA $p_k\in \mathcal{P}$ as:
\begin{equation}\label{eq:util1}
\begin{aligned}
&\hspace{-3mm}c^{p_k}(\{\mathbf{M}_{j,(a_k)}\}_{k=1}^{|\mathcal{P}|})  \triangleq \Pi_{i=1}^{n_d} \mathbf{1}_{\{L^i+\sum_{k=1}^{|\mathcal{P}|} m^{i}_{j,(a_k)} \leq |\mathcal{S}^i| \}} \Big( \rho_j- \\  
&\hspace{-3mm}\chi_j \Upsilon(\{M_{j,(a_k)}\}_{k=1}^{|\mathcal{P}|})  -\phi_j  \sum_{i=1}^{n_d} \mathbf{1}_{\{{m}^{i}_{j,(a_k)}>0\}}
+ \chi_j \Upsilon^{max} \\ 
&\hspace{-3mm}+\phi_j |\mathcal{V}_j|\Big) - \left(1- \Pi_{i=1}^{n_d} \mathbf{1}_{\{L^i+\sum_{k=1}^{|\mathcal{P}|} m^{i}_{j,(a_k)} \leq |\mathcal{S}^i| \}} \right) \Xi^{p_k},
\end{aligned}
\end{equation}
where $\Pi_{i=1}^{n_d} \mathbf{1}_{\{L^i+\sum_{k=1}^{|\mathcal{P}|} m^{i}_{j,(a_k)} \leq |\mathcal{S}^i| \}}$ ensures the availability of enough free slots, $\Xi^{p_k}$ denotes the penalty for delaying the execution, the constants are the same as Eq.~\eqref{eq:util2}, and $\Upsilon^{max}$ denotes the maximum payment of a PA. In this case, the utility of each PA depends not only on its own choice of action but also on the chosen actions by others. In this paradigm, we model the interactions between the PAs as a \textit{non-cooperative game}, more specifically a multi-player \textit{normal form game}. Consequently, we assume that each PA is rational in the sense that it aims to maximize its own utility function. In summary, for $Gjob_j$, the game can be defined as: $\mathbb{G}_j=(\mathcal{P},\{\mathcal{SA}_j\},\{c^p\}_{p \in \mathcal{P}})$, where $c^p$ is the utility of PA $p\in\mathcal{P}$. To solve this game, we use the concept of correlated equilibrium (CE), which generalizes the idea of Nash equilibrium to enable correlated strategy choices among the players. For the proposed game $\mathbb{G}_j$,
we define $\pi_j$ as the probability distribution over the joint strategy space $\Pi_{k=1}^{|\mathcal{P}|} \mathcal{SA_\textit{j}} = \mathcal{SA_\textit{j}} ^{|\mathcal{P}|}$. The set of correlated equilibria $\mathcal{CE}_{j}$ is the
convex polytope given by the following expression:\footnote{$\mathbf{M}_{j,({-p_k})}$ denotes the strategy of all the PAs except $p_k$.}
\begin{align}
&\hspace{-2mm}\mathcal{CE}_{j} = \Big\{\pi_j: \displaystyle\sum_{\mathbf{M}_{j,({-p_k})} \in \mathcal{SA_\textit{j}} ^{|\mathcal{P}|-1}} \pi_j(\mathbf{M}_{j,(a_k)},\mathbf{M}_{j,({-p_k})})\nonumber \\
&\hspace{-3mm}\big[c^{p_k}(\mathbf{M}_{j,({a_k'})},\mathbf{M}_{j,({-p_k})})- c^{p_k}(\mathbf{M}_{j,(a_k)},\mathbf{M}_{j,({-p_k})})\big]\leq 0,\nonumber\\
&\hspace{-2mm}\forall p_k \in \mathcal{P}, \forall \mathbf{M}_{j,(a_k)},\mathbf{M}_{j,({a_k'})} \in \mathcal{SA}_\textit{j} \Big\}.
\end{align}
Inspired by the pioneer work~\cite{ref:regretmatching}, we propose a distributed algorithm, called regret matching-based assignment (RMBA) algorithm, to solve the proposed game while reaching the CE. The PAs' actions are described in Algorithm~\ref{alg:corrEqui}. In RMBA algorithm, each PA saves the history of actions of the other PAs, using which he obtains the past rewards of the actions given that the other PAs would have taken the same actions (Eq.~\eqref{eq:sameac}). Afterward, each PA derives the regret of not executing different strategies (Eq.~\eqref{eq:regretmatch},~\eqref{eq:regretmatch1}). Finally, the probability of selection of the strategies are determined, where strategies with higher rewards in the past receive higher probabilities (Eq.~\eqref{eq:probEqui}).\footnote{Partitioning the time into ``time-frames'' does not have a significant impact on the convergence of the RMBA algorithm. This is due to the fact that at each iteration of this algorithm, the regret for all the strategies are obtained considering the previously taken action of all the PAs, which reduces the chance of trapping in low utility strategies at the initial iterations (see Remark~\ref{remark:Grouping}).}
\setlength{\textfloatsep}{0pt}
 \begin{algorithm}[h]
 {\footnotesize
 	\caption{RMBA: Regret matching-based assignment}\label{alg:corrEqui}
 	\SetKwFunction{Union}{Union}\SetKwFunction{FindCompress}{FindCompress}
 	\SetKwInOut{Input}{input}\SetKwInOut{Output}{output}
    \Input{PA $p_k$, graph job's type $j$, number of iterations $N_j$, $\mathcal{SA}_j$.}
    Select strategies randomly for the first iteration ($n=1$).\\
    \For{$n=1$ to $N_j$}{
    Denote history of allocations up to iteration $n$ as $\{\mathbf{A}_{\tau}\}_{\tau=1}^{n}$, where $\mathbf{A}_{t}$ is a vector consisting of $|\mathcal{P}|$ elements, which the $j^{th}$ one, $A_t^{j}$, corresponds to $p_j$'s used allocation at iteration $t$.\\
    Calculate the substituting reward for every two different allocations ($\forall m_1,m_2\in \mathcal{SA}_j$):
    \begin{equation}\label{eq:sameac}
    \begin{aligned}
    SR^{p_k}_{m_1,m_2}(n)=
    \begin{cases}
    c^{p_k}(m_2,\mathbf{A}_n^{-{k}})  &\textrm{  if  } A_n^{k}=m_1,\\
    c^{p_k}(\mathbf{A}_n) &\textrm{O.W}.
    \end{cases}
    \end{aligned}
    \end{equation}\\
    Calculate the substituting average rewards:
    \begin{equation}\label{eq:regretmatch}
    \Delta^{p_k}_{m_1,m_2}(n)= \frac{\sum_{\tau=1}^{n} SR^{p_k}_{m_1,m_2}(\tau)-c^{p_k}(\mathbf{A}_{\tau})}{n}.
    \end{equation}\\
    Calculate the average regret:
    \begin{equation}\label{eq:regretmatch1}
    R^{p_k}_{m_1,m_2}(n)=\max \{\Delta^{p_k}_{m_1,m_2}(n),0 \}.
    \end{equation}\\
    Form the selection probability distribution of the allocations, $\forall m \in \mathcal{SA}_j$, for the next iteration:
    \begin{align}
    \hspace{-6mm}
     \begin{cases}
      \hspace{-.0mm}
     p^{p_k}_{j,(m)}(n+1)\hspace{-.5mm}=\hspace{-.5mm} \frac{1}{\hspace{-0mm} \sum_{m'\in \mathcal{SA}_j} R^{p_k}_{A^{k}_{n},m'}(n)} R^{p_k}_{A^{k}_{n},m}(n) 		\\
      \hspace{48mm}\textrm{if } m\neq A^{k}_{n} \hspace{-0.2mm}, \\
      \hspace{-.5mm}
    p^{p_k}_{j,(m)}(n+1)\hspace{-.5mm}=\hspace{-.5mm} 1- \displaystyle\sum_{m'\in \mathcal{SA}_j: m'\neq m}  p^{p_k}_{j,({m'})}(n+1) &\\
    \hspace{48mm}\textrm{Otherwise}.
    \end{cases}\label{eq:probEqui}
    \end{align}\\
    Choose a strategy for the next iteration $A^{p_k}_{n+1}$ according to the derived distribution.
    }
    }
 \end{algorithm}

\section{Simulation Results}\label{sec:Simu}
\subsection{Simulation of a Small-Scale \GC}
\noindent In this scenario, the network consists of $5$ fully-connected DCs. The number of slots per DC is assumed to follow one of the three scenarios described in Table \ref{table:Slots}. Each of the cloud servers inside a DC is assumed to have $3$ slots.
\begin{table}[H]
\begin{center}
{\footnotesize
\begin{tabular}{|c | c| c| c| c| c|}
\hline
& DC\_1 &  DC\_2 &  DC\_3 &  DC\_4&  DC\_5 \\\hline 
Scenario $1$& $9$ &$12$ &$15$ &$18$ &$21$\\  \hline
Scenario $2$& $12$ &$15$ &$18$& $21$ &$24$\\  \hline
Scenario $3$& $15$& $18$ &$21$ &$24$ &$27$\\  \hline
\end{tabular}
}
\end{center}
\caption{Number of slots of each DC for different scenarios.}\label{table:Slots}
\end{table}
For all DCs, the following parameters are chosen according to~\cite{ref:powerInfo}, $P_{idle}=150W$, $P_{max}=250W$, $\eta=1.3$, and $\alpha=3$ (these numbers were reported to model the IBM BladeCenter server). Type and topology of the graph jobs are depicted in Fig.~\ref{fig:graphJobs}. It is assumed that at each iteration $10$ graph jobs are needed to be allocated. The arrival rates of graph jobs are set to $1$, $1$, $1$, $3$, and $4$ per iteration for type $1$ to type $5$.\footnote{{\color{black}It is observed that large graph jobs containing more nodes and complicated communication constraints lead to a larger performance gap between our proposed methods and the baseline algorithms as compared to small graph jobs. Note that the actual performance gap is dependent on the topology of the  graph job and may vary from one to another.}}
The initial load of each DC is assumed to be a random variable uniformly distributed between $0$ and $20\%$ of its number of slots. For each DC, inspired by~\cite{ref:powerDCsurvey}, $\nu$ is chosen to be $5\%$ of the peak power consumption. The cost of electricity ($\xi^i,\; \forall i$) is chosen to be the average cost of electricity in the US $0.12 \$/kWh$. In simulations, we use the term ``incurred power" referring to the difference between the power consumption of the {\GC} after the graph jobs are assigned as compared to that before the assignment. {\color{black} Since we are among the first to study the power-aware allocation of graph jobs in GDCNs, there is no baseline for direct comparison. Hence, we propose two greedy algorithms as the baselines:}

\textbf{Greedy 1:} In this algorithm, for each DC, its future power consumption upon allocating all the nodes of the arriving graph job to it is derived. Then, the DCs are sorted in ascending order according to their future power consumption as: $d^{i_1},\cdots,d^{i_{n_d}}$. Finally,  from the feasible set of assignments, the  assignment with the largest number of utilized slots from $d^{i_1}$ is chosen. The ties are broken considering the number of utilized slots from $d^{i_{2}}$ and so on.

\textbf{Greedy 2:} In this algorithm, at first the number of free slots of each DC is derived and the DCs are sorted in descending order with respect to their number of free slots as: $d^{i_1},\cdots,d^{i_{n_d}}$. From the feasible set of assignments, the assignment with the largest number of utilized slots from $d^{i_1}$ is chosen. Upon having a tie, the process continues by comparing the number of slots utilized form $d^{{i}_{2}}$ and so on.

Simulation results of the sequential graph job allocation described in Section~\ref{sec:Small} are presented in Fig.~\ref{diag:scenario1}. For the third scenario\footnote{The rest are omitted due to similarity.} described in Table~\ref{table:Slots}, Fig.~\ref{diag:scenario1}(\subref{diag:scenario1a}) reveals a negligible gap between solving the integer programming described in Eq.~\eqref{eq:int1}-\eqref{eq:int2} using exhaustive search and using the subsequent proposed sub-optimal method (Eq.~\eqref{eq:ms}-\eqref{eq:weightedmeanSquare}), and Fig.~\ref{diag:scenario1}(\subref{diag:scenario1b}) depicts the corresponding cumulative average incurred power of allocations. For all the three scenarios, Fig.~\ref{diag:scenario1}(\subref{diag:scenario1c}) depicts the cumulative profit obtained after $100$ iterations upon using the proposed sub-optimal approach as compared to the greedy methods if graph jobs stay busy in the system for $24$ hours at each round of allocation. As can be seen from Fig.~\ref{diag:scenario1}(\subref{diag:scenario1c}), on average, our method leads to saving of \$1100 on electricity cost.
	\begin{figure*}[htb!]
		\centering
		\subcaptionbox{Convergence of local parameter $\Lambda$ for different DCs.\label{diag:scenarioDista}}[.3\linewidth][c]{%
		\includegraphics[width=.330\linewidth]{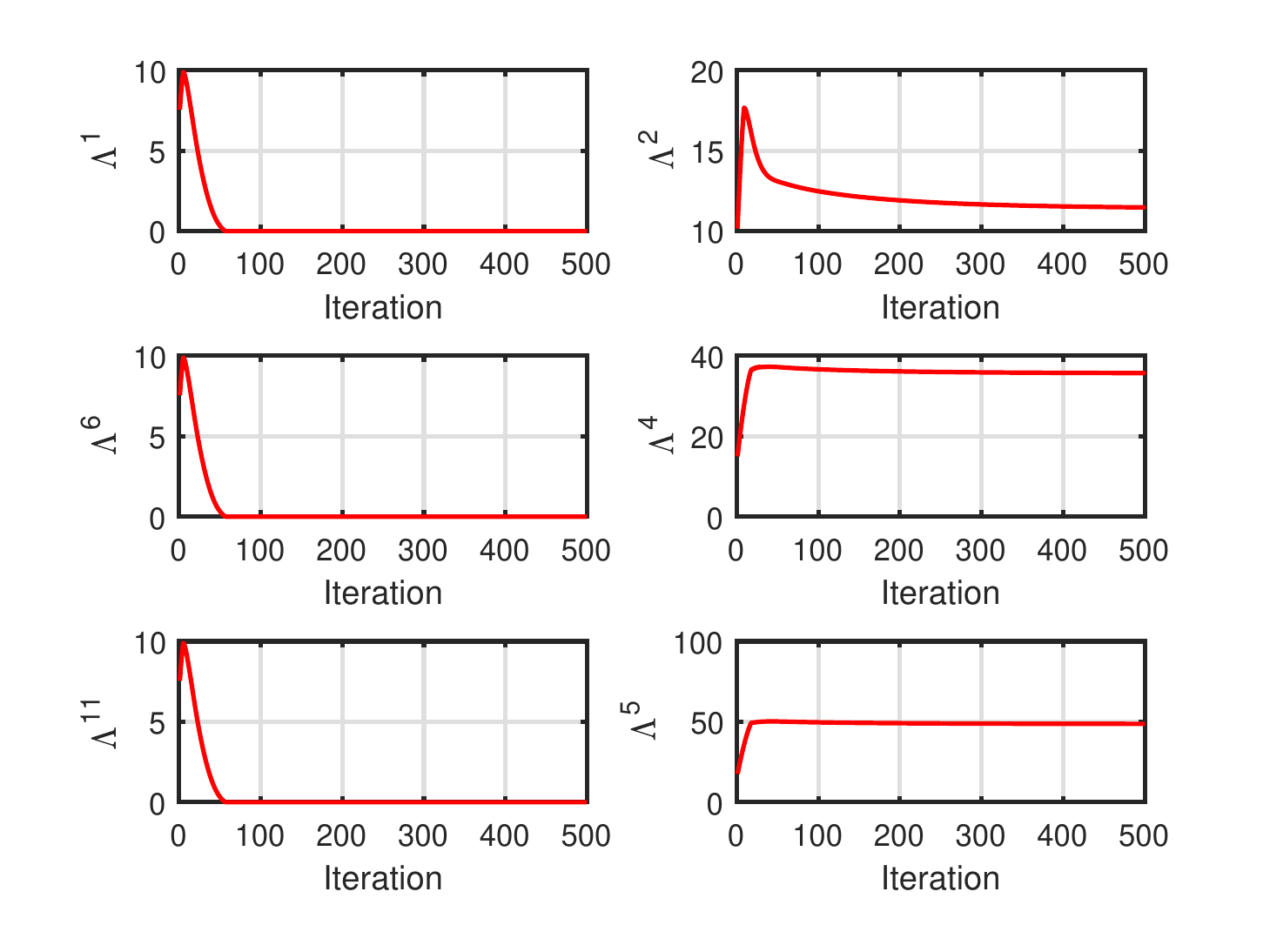}}\quad
		\subcaptionbox{Convergence of global parameter $\gamma$ for different DCs.\label{diag:scenarioDistb}}[.3\linewidth][c]{%
		\includegraphics[width=.330\linewidth]{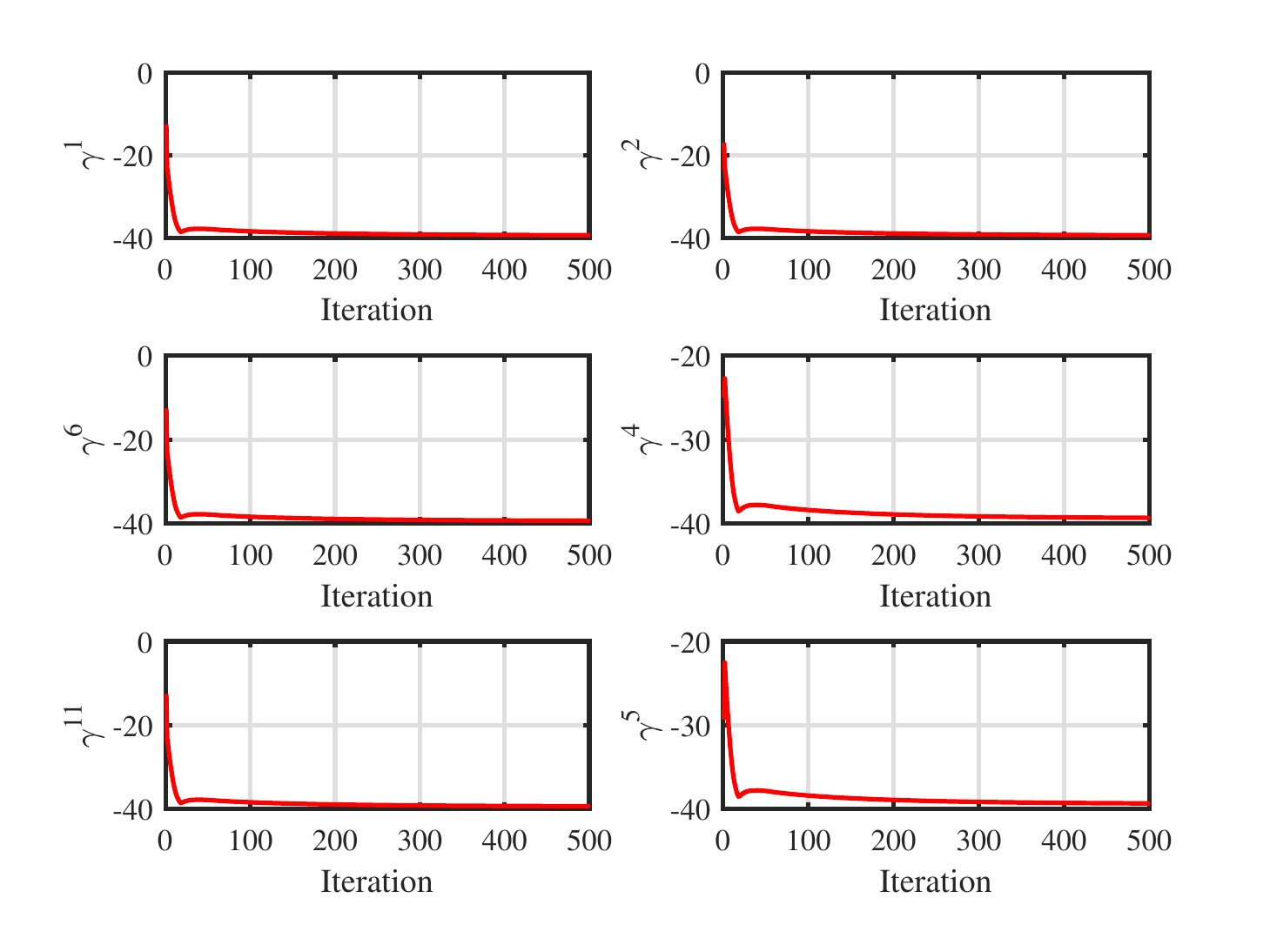}}\quad
		\subcaptionbox{Number of utilized slots of various DCs.\label{diag:scenarioDistc}}[.3\linewidth][c]{%
		\includegraphics[width=.330\linewidth]{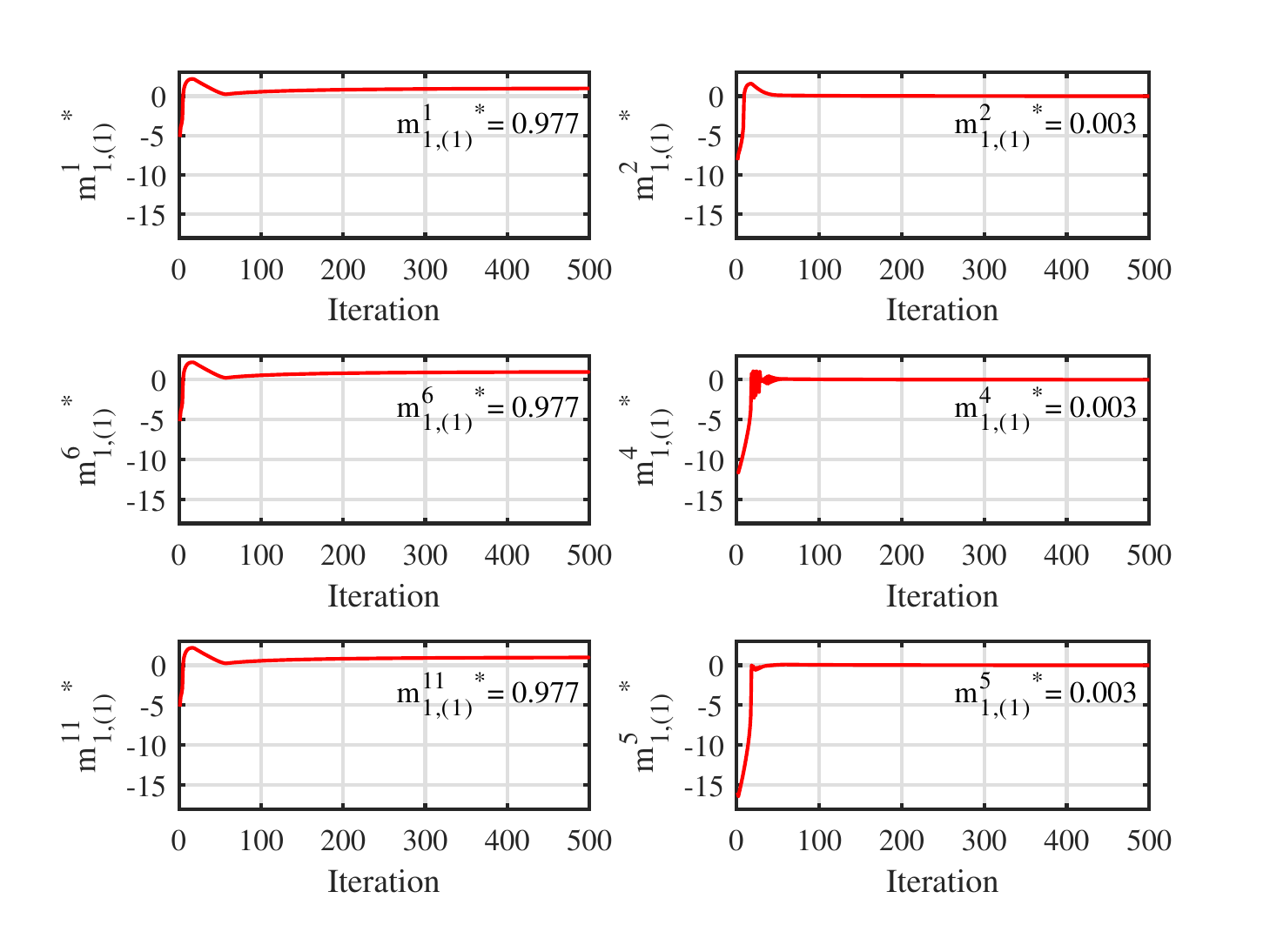}}
       
		\caption{Evolution of the parameters using the CDGA algorithm. The parameters of only 6 DCs are shown for more readability.}
         \label{diag:scenarioDist}
         \vspace{-4.5mm}
	\end{figure*}
	\begin{figure*}[htb!]
		\centering
		\subcaptionbox{The incurred power of allocation using the optimal allocation (exhaustive search), CDGA algorithm, and greedy algorithms.\label{diag:scenario22a}}[.3\linewidth][c]{%
		\includegraphics[width=.330\linewidth]{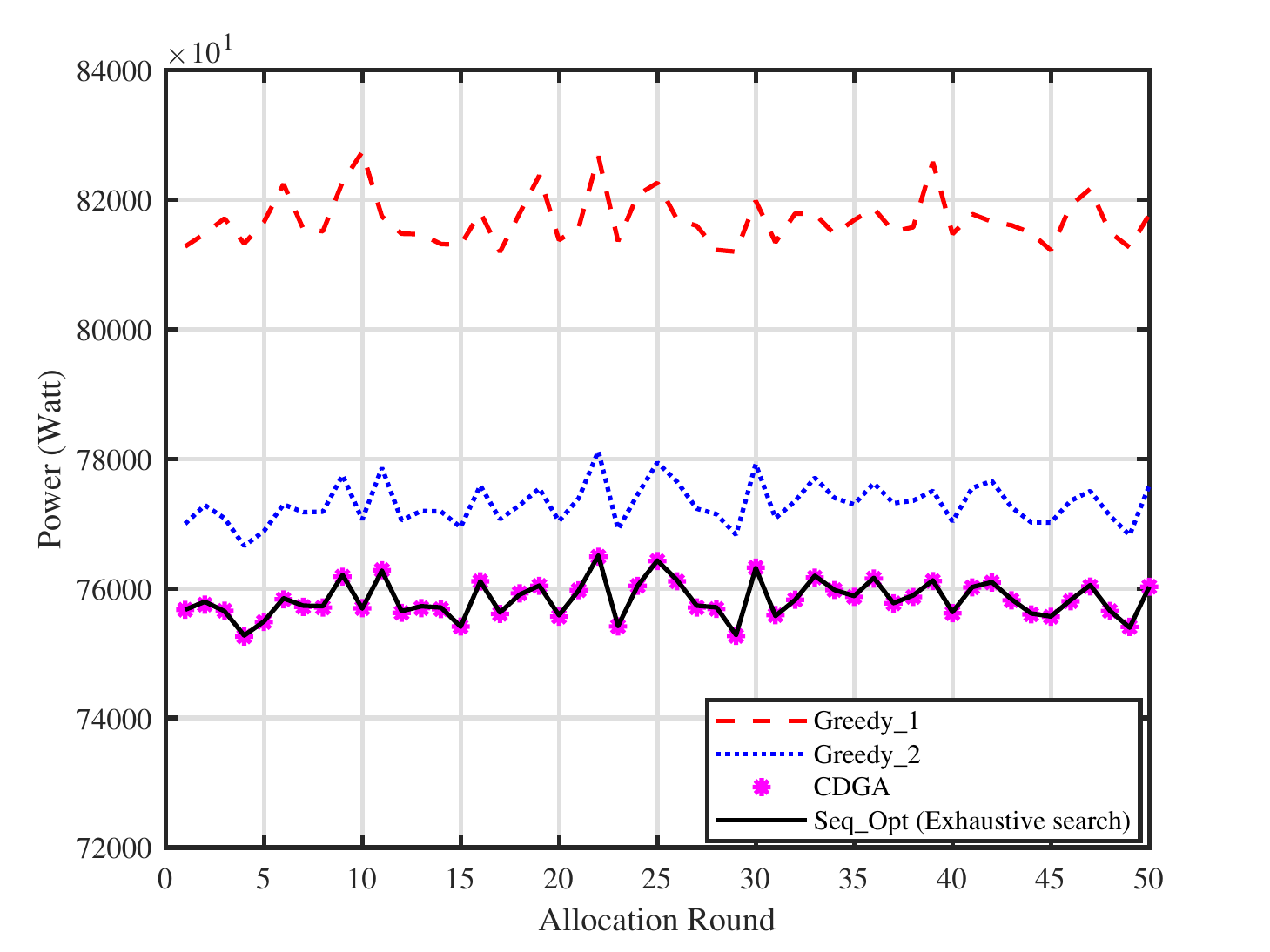}}\quad
		\subcaptionbox{Cumulative average incurred power of allocation using the optimal allocation (exhaustive search), CDGA algorithm, and greedy algorithms.\label{diag:scenario22b}}[.3\linewidth][c]{%
		\includegraphics[width=.330\linewidth]{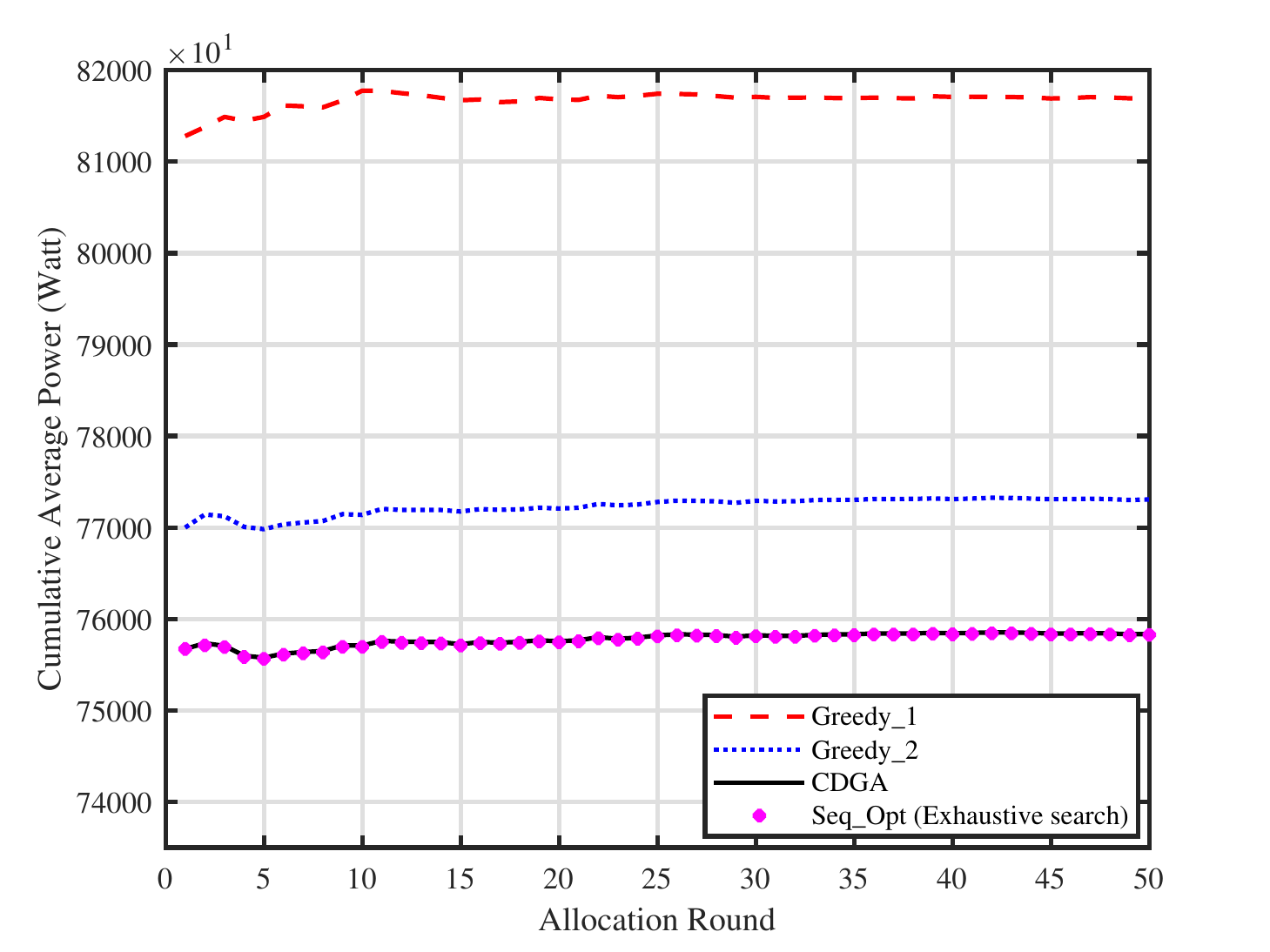}}\quad
		\subcaptionbox{The generated revenue upon using the CDGA algorithm as compared to the greedy methods.\label{diag:scenario22c}}[.3\linewidth][c]{%
		\includegraphics[width=.330\linewidth]{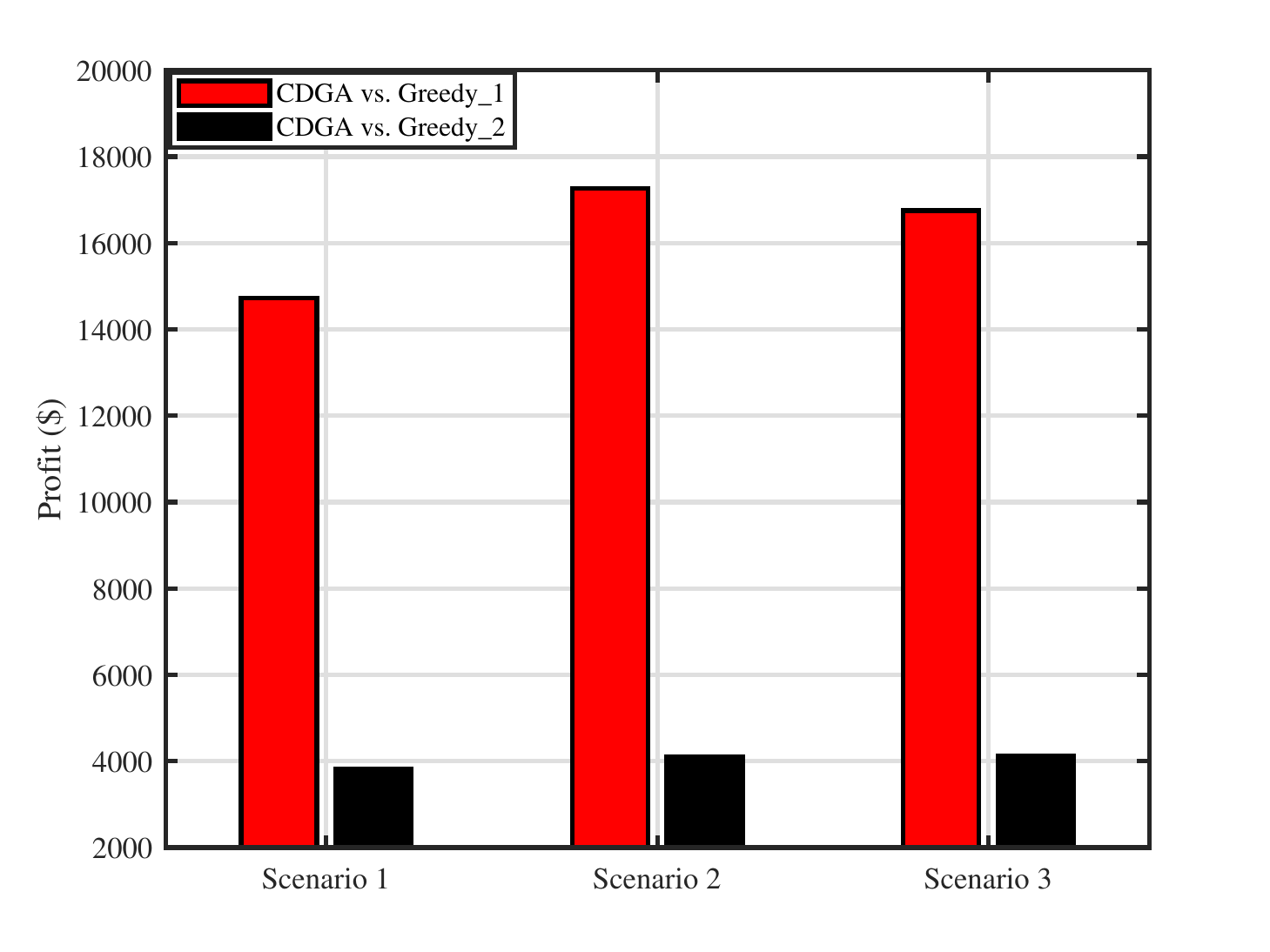}}
       
		\caption{Simulation results of the medium-scale {{\GC}}.}
         \label{diag:scenario22}
                  \vspace{-4.5mm}
	\end{figure*}
\subsection{Simulation of a Medium-Scale \GC}
In this scenario, the network comprises $15$ fully-connected DCs. Similar to the previous case, we consider $3$ scenarios, each created via (fully) connecting the three replicas of a {\GC} described in Table~\ref{table:Slots}. The number of slots per DC, power parameters of each DC, $\nu$, and the price of electricity are assumed to be the same as the previous case. It is assumed that at each iteration, there are $30$ graph jobs waiting to be assigned to the network (the arrival rate for each type is three times higher than that of the first scenario). The step-sizes in Eq.~\eqref{eq:full_dist_2} are set as follows: $c_{\lambda}=0.1$, $c_{\gamma}=0.18$, and $c_{\Lambda}=0.15$. We choose $\phi=5$, and for deriving the matrix $\mathbf{W}$, the parameter $\epsilon$ is chosen to be $0.1$. Also, we choose the initial value of $\Lambda^i$, $\lambda^i$, $\gamma^i$ at DC $d^i$ to be $ \frac{\eta^i N^i  \sigma^i \alpha^i}{5(|\mathcal{S}^i|)^{\alpha^i}}$, $0$, and $\frac{-\eta^i N^i  \sigma^i \alpha^i}{3(|\mathcal{S}^i|)^{\alpha^i}}$, respectively. For assignment of a graph job with type $1$, Fig.~\ref{diag:scenarioDist} depicts the convergence of the local and global variables of the DCs at $6$ sampled DCs. Fig.~\ref{diag:scenarioDist}(\subref{diag:scenarioDista}) describes the convergence of the local variable $\Lambda$ at the DCs, Fig.~\ref{diag:scenarioDist}(\subref{diag:scenarioDistb}) shows the convergence of the global variable $\gamma$, and Fig.~\ref{diag:scenarioDist}(\subref{diag:scenarioDistc}) depicts the number of offered slots of each DC to the graph job. The parameter $\lambda$ takes the value of zero almost always upon the convergence, and thus is not depicted. As can be seen from Fig.~\ref{diag:scenarioDist}(\subref{diag:scenarioDistc}), the DCs $1$, $6$, and $11$ would offer one slot to the graph job while the rest would offer zero slots. Also, from Fig.~\ref{diag:scenarioDist}(\subref{diag:scenarioDistb}), it can be seen that the initial non-identical choices of the global variable $\gamma$ at each DC finally converges to a unified value among the DCs. 
	
Fig.~\ref{diag:scenario22} depicts the results of comparing the CDGA algorithm to the greedy algorithms. Fig.~\ref{diag:scenario22}(\subref{diag:scenario22a}) depicts the incurred power of allocation for scenario $3$ of network construction, and Fig.~\ref{diag:scenario22}(\subref{diag:scenario22b}) shows the cumulative average incurred cost of allocation. These two figures reveal the close-to-optimal performance of the CDGA algorithm. Also, Fig.~\ref{diag:scenario22}(\subref{diag:scenario22c}) demonstrates the obtained profit using the CDGA algorithm as compared to the greedy algorithms after $100$ iterations. As can be seen from this figure, on average, our method results in saving of \$10000 on electricity cost.

\begin{figure*}[t]
		\centering
		\subcaptionbox{Comparison between the $P^{90\%}$ of the BRMA algorithm and the random strategy selection for various graph jobs.\label{diag:fixedPricesa}}[.3\linewidth][c]{%
		\includegraphics[width=.330\linewidth]{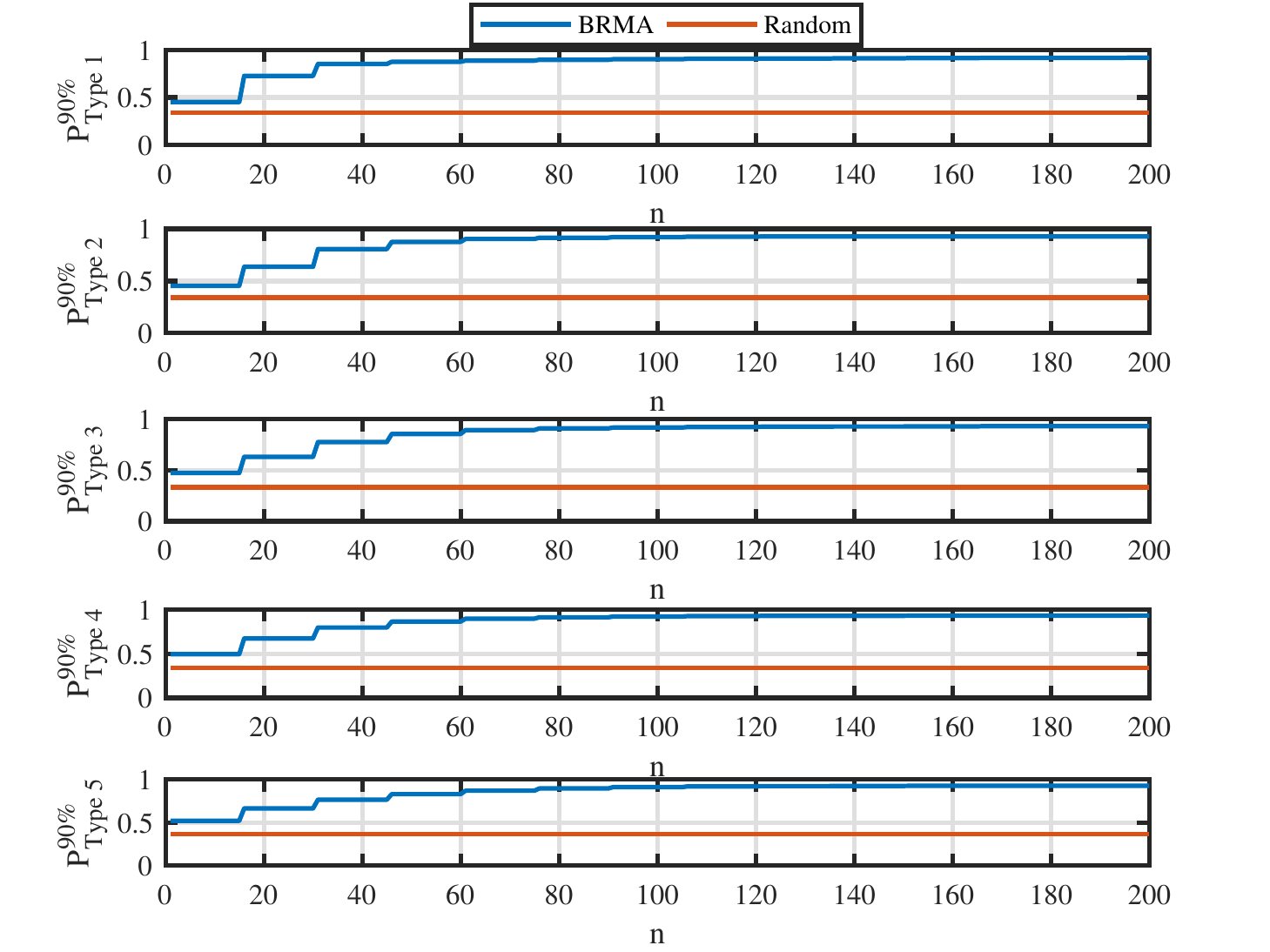}}\quad
		\subcaptionbox{Comparison between the utility of the BRMA algorithm and the random strategy selection for various graph jobs.\label{diag:fixedPricesb}}[.3\linewidth][c]{%
		\includegraphics[width=.330\linewidth]{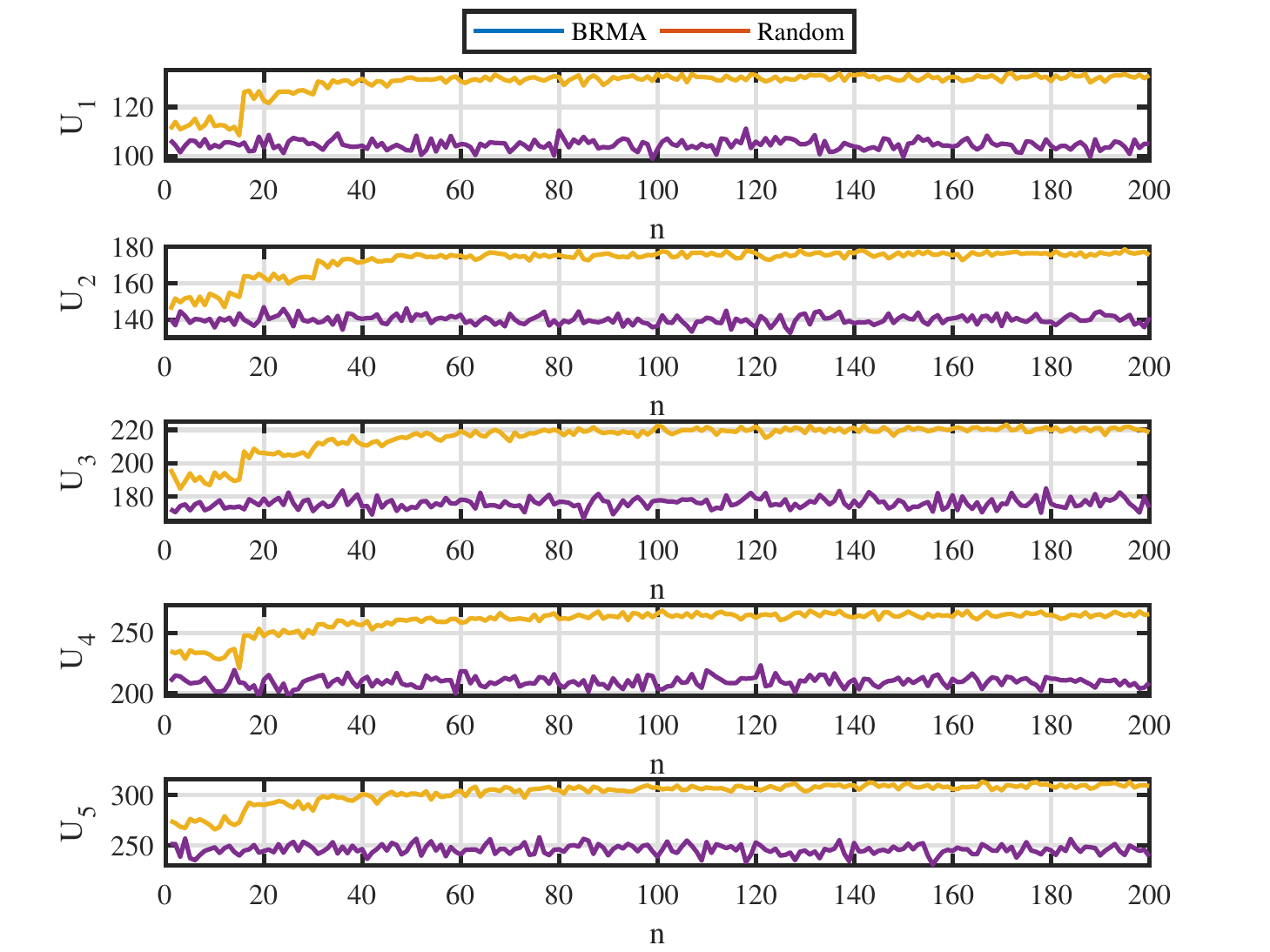}}\quad
		\subcaptionbox{Effect of having different similarity considerations (left) and different number of clusters (right) for graph job with type $5$.\label{diag:fixedPricesc}}[.3\linewidth][c]{%
		\includegraphics[width=.330\linewidth]{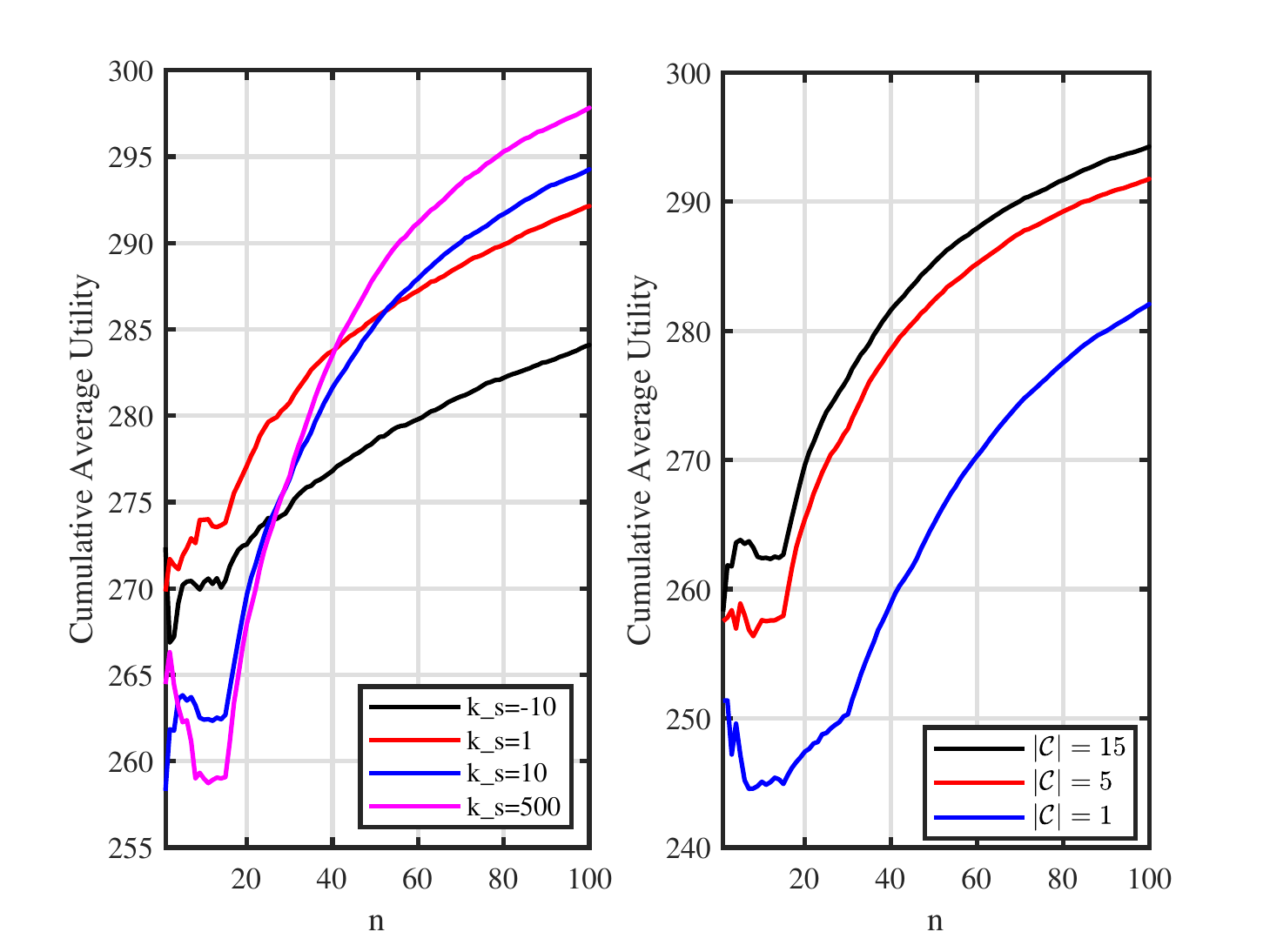}}
		\caption{Simulation results of the large-scale {{\GC}} upon having DCs with fixed prices.}
        \label{diag:fixedPrices}
                 \vspace{-4mm}
       \end{figure*}

        \begin{figure*}[t]
        \vspace{3mm}
		\centering
		\subcaptionbox{Average utility of all the PAs for different types of graph jobs (left), and utility of individual PAs for each type of graph job in one round of RMBA algorithm (right).\label{diag:dynamicPricesa}}[.3\linewidth][c]{%
		\includegraphics[width=.320\linewidth]{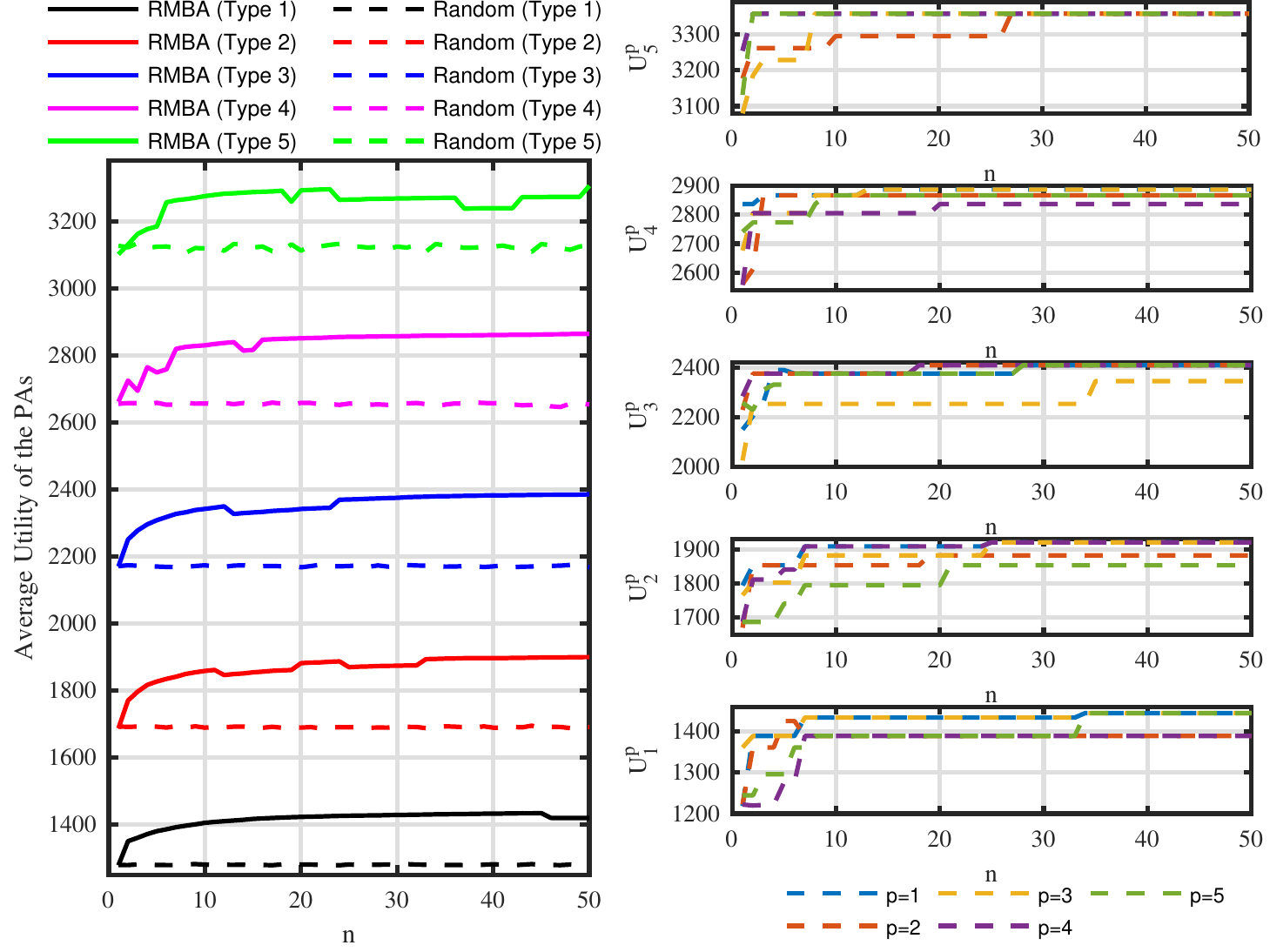}}\quad
		\subcaptionbox{Total power consumption of the utilized servers for the RMBA algorithm and random strategy selection.\label{diag:dynamicPricesb}}[.3\linewidth][c]{%
		\includegraphics[width=.330\linewidth]{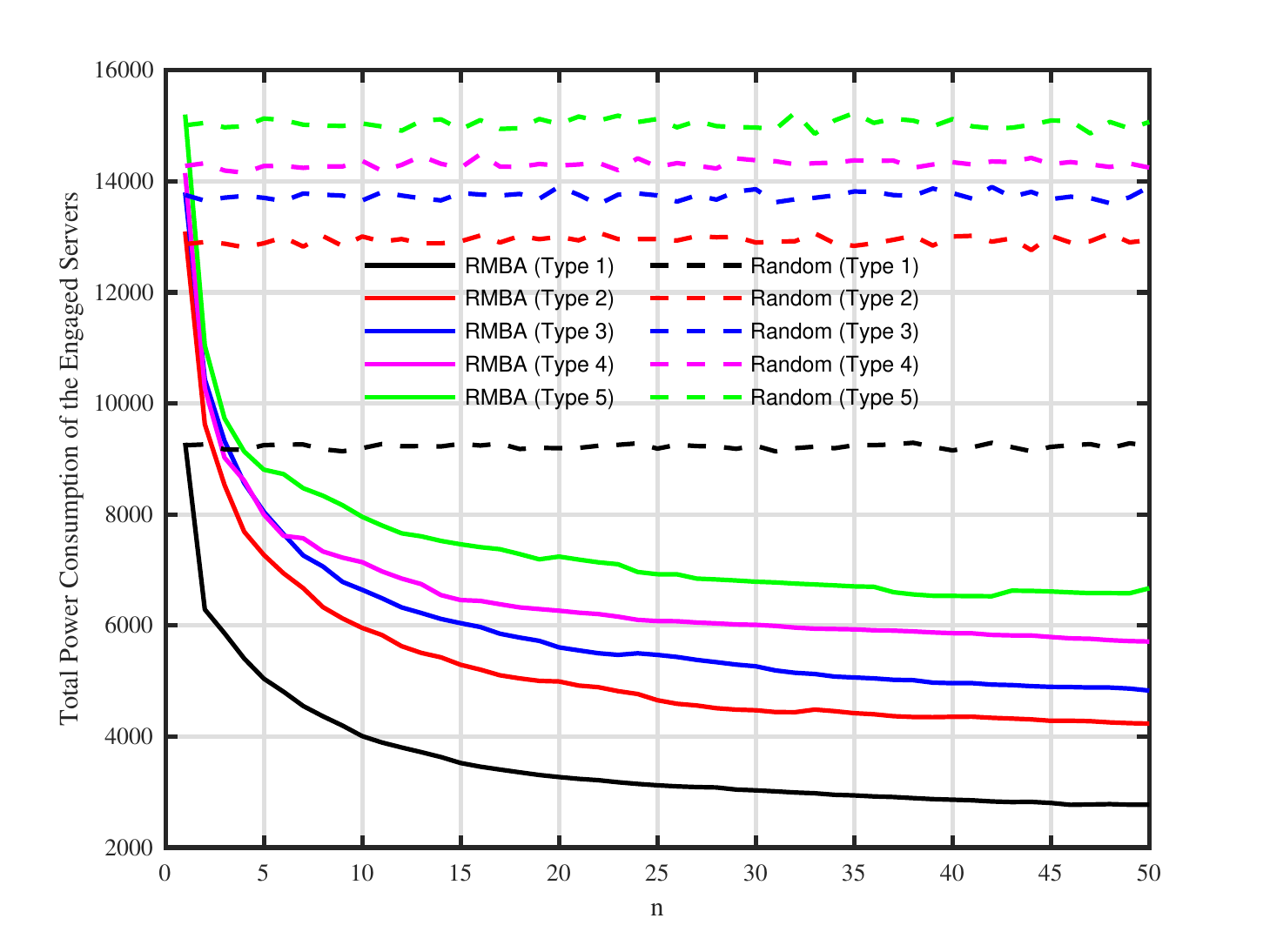}}\quad
		\subcaptionbox{Reduction in currency in circulation upon using the RMBA algorithm as compared to the random strategy selection.\label{diag:dynamicPricesc}}[.3\linewidth][c]{%
		\includegraphics[width=.330\linewidth]{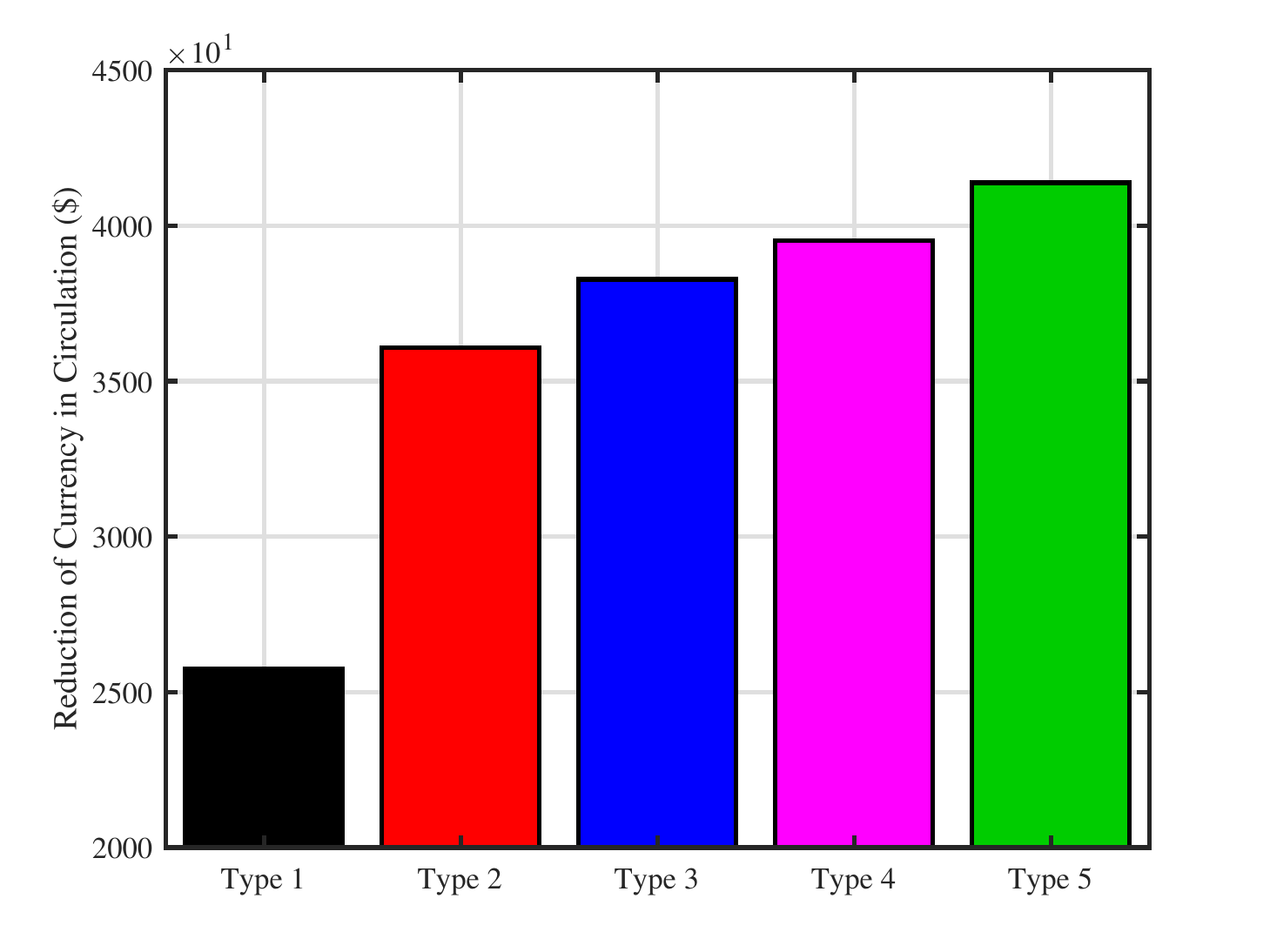}}
       
		\caption{Simulation results of the large-scale {{\GC}} upon having DCs with dynamic prices.}
         \label{diag:dynamicPrices}
         \vspace{-4mm}
	\end{figure*}
\subsection{Simulation of a Large-Scale \GC}
\subsubsection{The Network and the CCR's Setting} We consider a {\GC} consisting of 200 DCs each of which possesses $3k$ slots, where $k$ is a random integer number between $4$ and $11$. It is assumed that each server in a DC contains $3$ slots. The power parameters of each DC, $\nu$, and the price of electricity are assumed to be the same as before. In \cite{gyarmati2010scafida}, a scale-free (SF) architecture called Scafida is proposed for DC networks. The main advantages of this model are error tolerance, scalability, and flexibility of network architecture. Also, this graph structure is used to model Internet connections \cite{ref:SF}. Considering these facts, the network topology of the {\GC} is assumed to be an SF graph constructed using preferential attachment mechanism with parameter $m = 3$~\cite{ref:SF}. For each type of graph job, we run a CCR on the network, where the initial place of the CCR is randomly chosen among the DCs and the size of carrying BST is set to $1000$. In the following, the strategies of the PAs are always chosen from the suggested strategies of the CCRs.

\subsubsection{DCs with Fixed Prices}
In this case, the load of each DC is assumed to be a uniform random variable between $20\%$ and $100\%$ of its number of slots. The price of utilizing each slot is fixed to be the total electricity cost of the DC divided by its number of slots. For a PA, for each type of job, the utility is derived by fixing all the values of $\rho$, $\chi$, and $\phi$ to $1$. The strategy exploration parameter $E$ is set to $0.01$, length of time-frame $\Gamma$ is chosen to be $15$, $|\mathcal{C}|=15$, $k_s=10$ in Eq.~\eqref{eq:Q2}, and $K=10$ in Eq.~\eqref{eq:weight}. To increase the convergence rate, we force the algorithm to execute one strategy from each cluster during the first time-frame. The resulting curves are obtained via averaging over $200$ simulations. In this case, our BRMA algorithm is compared with the random strategy selection method due to the lack of an existing algorithm for the problem. Note that in this case, load parameters and prices of the DCs are unknown prior to the graph job assignment; hence the above-defined greedy algorithms can not be applied to this setting. Besides the utility, we also introduce $P^{90\%}$ as a performance metric, which is the probability of selecting an allocation that has at most $10\%$ lesser utility as compared to the best suggested allocation. The results are depicted in Fig.~\ref{diag:fixedPrices}. With one graph job assigned per iteration, Fig.~\ref{diag:fixedPrices}(\subref{diag:fixedPricesa}) depicts the convergence of $P^{90\%}$ using the BRMA algorithm for all types of jobs, while Fig.~\ref{diag:fixedPrices}(\subref{diag:fixedPricesb}) reveals the corresponding utility of assignment as compared to the random strategy selection. As can be seen, the probability of choosing the high utility strategies increases while the BRMA algorithm explores the suggested strategies. Also, it can be seen that the proposed BRMA algorithm has a higher utility even at the first $15$ iterations. This is because at the first time-frame the BRMA algorithm exercises one strategy from each cluster, based on which strategies with good utilities are explored with higher probabilities. The utility gain of at least $20\%$ for each graph job assignment upon using our BRMA algorithm can be seen from Fig.~\ref{diag:fixedPrices}(\subref{diag:fixedPricesb}). Also, Fig.~\ref{diag:fixedPrices}(\subref{diag:fixedPricesc}) depicts the effect of changing the similarity coefficient $k_s$ (left sub-plot) and the effect of choosing various number of clusters (right sub-plot). For $k_s=-10$ all the strategies have a high similarity factor, and thus the algorithm does not perform well. As this parameter increases, the effect of similarity on weight update decreases, and the algorithm aims to select the best strategy with the highest utility rather than selecting a group of high utility strategies. Due to this fact, choosing $k_s=1$ as compared to $k_s=500$ results in a higher initial utility since the algorithm has more tendency toward choosing a portion of strategies with high utilities. However, choosing $k_s=500$ leads to a higher final utility. Also, as can be seen, having more clusters (up to the size of the time-frame) leads to a finer grain partitioning and a better performance.

 \subsubsection{DCs with Dynamic Prices}
 In this case, the PAs' payments is associated with the load of their utilized DCs (Eq.~\eqref{eq:loadprice}). The initial loads of the DCs are chosen to be uniformly distributed between the $20\%$ and $100\%$ of their number of slots. The presence of $5$ PAs is assumed in the system. For each PA, for each type of job, the utility is derived by fixing all the values of $\rho$, $\chi$, and $\phi$ to $1$ in Eq.~\eqref{eq:util1}. The results are presented in Fig.~\ref{diag:dynamicPrices}. In Fig.~\ref{diag:dynamicPrices}(\subref{diag:dynamicPricesa}), the average utility of the PAs for assigning one graph job over 100 Monte-Carlo simulations is presented (left plot). As can be seen, at least $20\%$ utility gain is obtained using our RMBA algorithm. To demonstrate the real-time performance of the RMBA algorithm, in the right plot of Fig.~\ref{diag:dynamicPrices}(\subref{diag:dynamicPricesa}), the utility of all the PAs for each type of graph job in one round of Monte-Carlo simulation is depicted. As can be seen, after exploring the environment and the other PAs' actions, each PA identifies the more rewarding strategies. Also, due to the inherent relationship between the utility of the PAs and the power consumption of the DCs, this method leads to less power consumption of DCs. This fact is revealed in Fig.~\ref{diag:dynamicPrices}(\subref{diag:dynamicPricesb}), where the corresponding power consumption associated with the utilized DCs opted by all the PAs is depicted for the RMBA algorithm and the random selection strategy. Fig.~\ref{diag:dynamicPrices}(\subref{diag:dynamicPricesc}) depicts the reduction in \textit{currency in circulation} obtained through less money gathering from the PAs and less payment for the electricity using the RMBA algorithm as compared to the random strategy selection.

\section{Conclusion}\label{sec:concl}
\noindent In this work, we study the problem of graph job allocation in geo-distributed cloud networks ({\GC}s). The slot-based quantization of the resources of the DCs is considered. Inspired by big-data driven applications, it is considered that tasks are composed of multiple sub-tasks, which need multiple slots of the DCs with a determined communication pattern. The cost-effective graph job allocation in {\GC}s is formulated as an integer programming problem. For small-scale {\GC}s, given the feasible assignments of the graph jobs, we propose an analytic sequential sub-optimal solution to the problem. For medium-scale {\GC}s, we introduce a distributed algorithm using the communication infrastructure of the network. Given the impracticality of those methods in large-scale {\GC}s, we propose a decentralized graph job allocation framework based on the idea of strategy suggestion using our introduced cloud crawlers (CCRs). To opt efficient strategies from the pool of suggested strategies, we propose two online learning algorithms for the PAs considering fixed and adaptive pricing of DCs. Extensive simulations are conducted to reveal the effectiveness of all the proposed algorithms in {\GC}s with different scales. For the future work, we suggest studying graph jobs with heterogeneous order of nodes' execution. {\color{black}Also, encapsulating the mathematical model of network link outages into the allocation of graph jobs among multiple DCs is worth further investigation.}

\bibliographystyle{IEEEtran}
\bibliography{BibCloud}
\vspace{0mm}
    \begin{IEEEbiography}[{\includegraphics[width=1.08in,height=1.25in,clip]{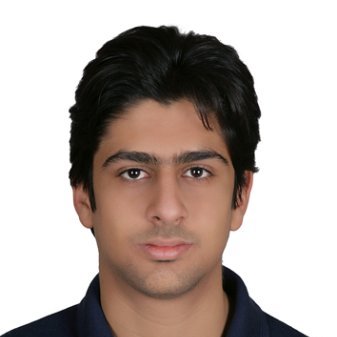}}]{Seyyedali Hosseinalipour} (S’17) received the B.S. degree in Electrical Engineering from Amirkabir University of Technology (Tehran Polytechnic), Tehran, Iran in 2015. He is currently pursuing a Ph.D. degree in the Department of Electrical and Computer Engineering at North Carolina State University, Raleigh, NC, USA. His research interests include analysis of wireless networks, resource allocation and load
balancing for cloud networks, and analysis of vehicular ad-hoc networks.
\end{IEEEbiography}

\begin{IEEEbiography}[{\includegraphics[width=1.12in,height=1.25in,clip]{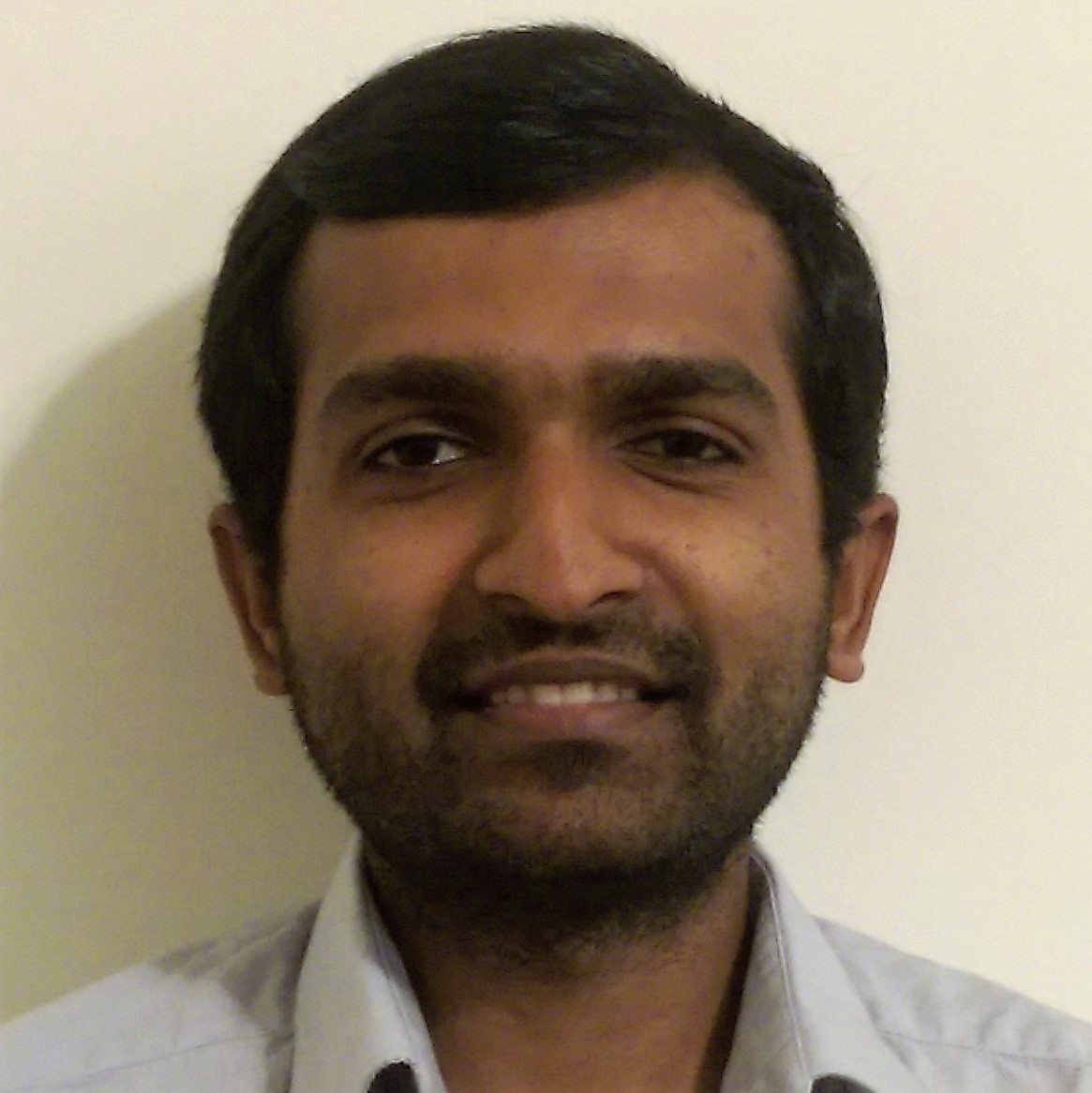}}]{Anuj Nayak} received his B. E. degree in Electronics and Communication Engineering from PES Institute of Technology, Bengaluru, India in 2014. He worked with Signalchip Innovations Pvt. Ltd., India as an algorithm design engineer from 2014 to 2016. He joined Master of Science in Electrical Engineering at North Carolina State University in 2016. His research interests are in the area of complex networks, statistical signal processing and artificial intelligence.

\end{IEEEbiography}
\begin{IEEEbiography}[{\includegraphics[width=1.12in,height=1.25in,clip]{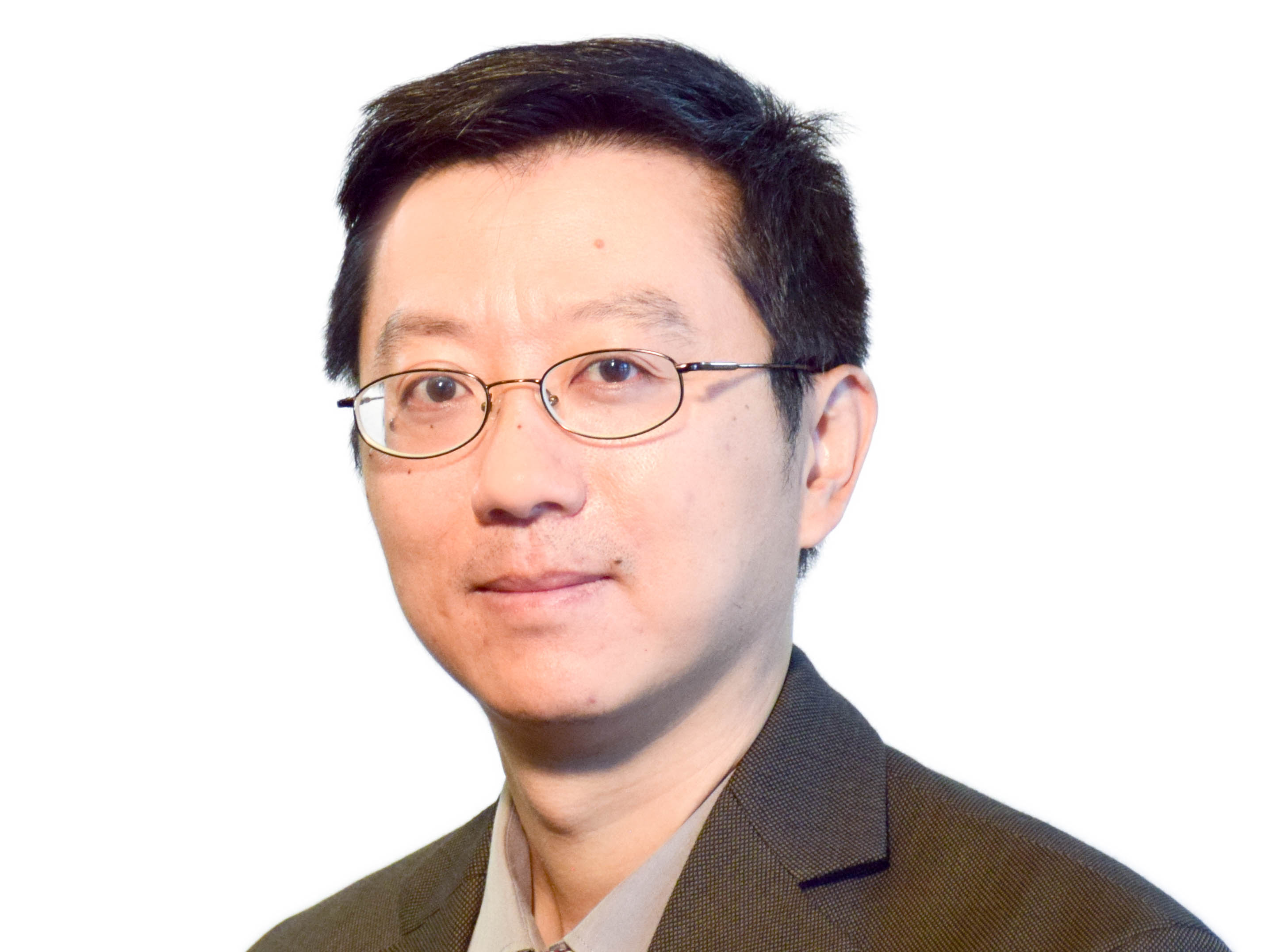}}]{Huaiyu Dai} (F’17) received the B.E. and M.S. degrees in electrical engineering from Tsinghua University, Beijing, China, in 1996 and 1998, respectively, and the Ph.D. degree in electrical engineering from Princeton University, Princeton, NJ in 2002. 

He was with Bell Labs, Lucent Technologies, Holmdel, NJ, in summer 2000, and with AT\&T Labs-Research, Middletown, NJ, in summer 2001. He is currently a Professor of Electrical and Computer Engineering with NC State University, Raleigh, holding the title of University Faculty Scholar. His research interests are in the general areas of communication systems and networks, advanced signal processing for digital communications, and communication theory and information theory. His current research focuses on networked information processing and crosslayer design in wireless networks, cognitive radio networks, network security, and associated information-theoretic and computation-theoretic analysis.  

He has served as an editor of IEEE Transactions on Communications, IEEE Transactions on Signal Processing, and IEEE Transactions on Wireless Communications. Currently he is an Area Editor in charge of wireless communications for IEEE Transactions on Communications. He co-chaired the Signal Processing for Communications Symposium of IEEE Globecom 2013, the Communications Theory Symposium of IEEE ICC 2014, and the Wireless Communications Symposium of IEEE Globecom 2014. He was a co-recipient of best paper awards at 2010 IEEE International Conference on Mobile Ad-hoc and Sensor Systems (MASS 2010), 2016 IEEE INFOCOM BIGSECURITY Workshop, and 2017 IEEE International Conference on Communications (ICC 2017).
\end{IEEEbiography}

    \end{document}